\documentclass{iopjournal} 

\usepackage{cancel}
\usepackage{amsfonts,amsmath,amssymb}
\usepackage{calc}
\usepackage[english]{babel}
\usepackage[normalem]{ulem}
\usepackage{graphicx}
\usepackage{subcaption,caption}
\usepackage{epsfig}
\usepackage{ulem}
\usepackage{tikz}
\usepackage{mathrsfs}
\usepackage{bbold}             
\usepackage{soul}

\usepackage{ragged2e}
\justifying

\usepackage{amsthm}

\theoremstyle{definition}

\theoremstyle{theorem}

\newtheorem{lemma}{Lemma}
\newtheorem{proposition}{Proposition}

\newtheorem{theorem}{Theorem}
\newtheorem{corollary}{Corollary}
\newtheorem{definition}{Definition}

\theoremstyle{definition}

\theoremstyle{remark}
\newtheorem{remark}{Remark}

\newcommand{\Ac}{\mathcal{A}}
\newcommand{\Bc}{\mathcal{B}}

\newcommand{\Dc}{\mathcal{D}}
\newcommand{\Ec}{\mathcal{E}}

\newcommand{\Hc}{\mathcal{H}}
\newcommand{\Ic}{\mathcal{I}}
\newcommand{\Jc}{\mathcal{J}}

\newcommand{\Lc}{\mathcal{L}}

\newcommand{\Pc}{\mathcal{P}}
\newcommand{\Qc}{\mathcal{Q}}
\newcommand{\Rc}{\mathcal{R}}
\newcommand{\Sc}{\mathcal{S}}
\newcommand{\Tc}{\mathcal{T}}
\newcommand{\Uc}{\mathcal{U}}

\newcommand{\As}{\mathscr{A}}

\newcommand{\Cs}{\mathscr{C}}
\newcommand{\Ds}{\mathscr{D}}

\newcommand{\Fs}{\mathscr{F}}

\newcommand{\Cb}{\mathbb{C}}

\newcommand{\Mb}{\mathbb{M}}
\newcommand{\Nb}{\mathbb{N}}

\newcommand{\Rb}{\mathbb{R}}

\newcommand{\Zb}{\mathbb{Z}}

\newcommand{\Bf}{\mathfrak{B}}

\newcommand{\Df}{\mathfrak{D}}

\newcommand{\um}{\frac{1}{2}}
\newcommand{\one}{\mathbb{1}}
\newcommand{\zero}{\mathbb{0}}

\newcommand{\tr}{{\rm tr}}
\newcommand{\diag}{{\rm diag}}
\newcommand{\Span}{{\rm span}}

\newcommand{\supp}{{\rm supp}}

\DeclareMathOperator{\rank}{rank}

\newcommand{\ket}[1]{\left| #1 \right>}
\newcommand{\bra}[1]{\left< #1 \right|}
\newcommand{\norm}[1]{\left|\left| #1 \right|\right|}
\newcommand{\inner}[2]{\left< #1,#2 \right>}

\newcommand{\expect}[1]{\left< #1 \right>}

\newcommand{\ketbra}[2]{\left\vert #1 \middle>\middle< #2 \right\vert}

\newcommand{\kket}[1]{\left|\left| #1 \right>\right>}

\definecolor{darkviolet}{rgb}{0.58, 0.0, 0.83}
\definecolor{lavender}{rgb}{0.45, 0.31, 0.59}

\begin{document}

\title{Approximate Reduced Lindblad Dynamics \\
via Algebraic and Adiabatic Methods}

\author{Tommaso Grigoletto$^1$, Alain Sarlette$^{2,3}$, Francesco Ticozzi$^{1,4,5}$, Lorenza Viola$^{4,*}$}

\affil{$^1$Department of Information Engineering, University of Padova, Italy} 

\affil{$^2$Laboratoire de Physique de l'Ecole Normale Sup\'erieure, Mines Paris-PSL, Inria, ENS-PSL, Universit\'e PSL, \\
\ CNRS, Sorbonne Universit\'e, Paris, France}

\affil{$^3$Department of Electronics and Information Systems, Ghent University, Belgium}

\affil{$^4$Department of Physics and Astronomy, Dartmouth College, Hanover, New Hampshire 03755, USA}

\affil{$^5$Padua Quantum Technologies Research Center, University of Padova, Italy}

\affil{$^*$ Corresponding author. Email: Lorenza.Viola@Dartmouth.edu}
\medskip


\begin{abstract} 
\vspace*{-3mm}
\justify{\small 
We present an algebraic framework for approximate model reduction of Markovian open quantum dynamics that guarantees complete positivity and trace preservation by construction. First, we show that projecting a Lindblad generator on its center manifold -- the space spanned by eigenoperators with purely imaginary eigenvalue -- yields an asymptotically exact reduced quantum dynamical semigroup whose dynamics is unitary, with exponentially decaying transient error controlled by the generator's spectral gap. Second, for analytic perturbations of a Lindblad generator with a tractable center manifold, we propose a perturbative reduction that keeps the reduced space fixed at the unperturbed center manifold. The resulting generator is shown to remain a valid Lindbladian; explicit finite-time error bounds that quantify leakage from the unperturbed center sector are provided. We further clarify the connection to adiabatic elimination methods, by both showing how the algebraic reduction can be directly related to a first-order adiabatic-elimination and by providing sufficient conditions under which the latter method can be applied while preserving complete positivity. We showcase the usefulness of our techniques in dissipative many-body quantum systems exhibiting non-stationary long-time dynamics. }
\end{abstract}

\section{Introduction}
\smallskip

Markovian models for open quantum systems provide a paradigmatic setting for describing irreversible continuous-time quantum dynamics \cite{alicki-lendi}. As such, they are widely used to account for the influence of a memory-less environment in a variety of both natural and engineered physical platforms of relevance to quantum science and technology. Markovian (semigroup) master equations, governed by a generator in canonical Lindblad form (also referred to as a {\em Lindbladian} henceforth), have long been used to describe noise effects in fields ranging from quantum optics \cite{Cohen} to atomic and condensed-matter physics \cite{Weiss, Jamir}, as well as to design open dynamics achieving non-unitary quantum-control tasks such as autonomous state stabilization  \cite{Ticozzi2007QuantumMS,ticozzi2013steadystateentanglementengineeredquasilocal,Cirac,Murch}.

However, even while not explicitly accounting for the environmental degrees of freedom, the resulting dynamical model may still be too large to simulate, or otherwise too complex to interpret and use for the purpose at hand. Notably, the rise of new research areas including open quantum many-body systems and driven-dissipative quantum matter \cite{Fazio, Sieberer}, along with the possibility to push experimental investigations in increasingly more powerful quantum simulators \cite{Altman}, heighten the need for smaller-dimensional, ``reduced'' descriptions that retain a valid quantum nature and reliably capture the dynamics of degrees of freedom of interest, without carrying the full complexity. A compelling scenario of interest arises, in particular, in regimes where dissipation does not relax the system merely to a steady state, but to a richer asymptotic structure known as the {\em center manifold} and tied to the occurrence of purely imaginary eigenvalues of the Lindblad generator \cite{Schirmer2010}. In such a case, {\em non-stationary long-time dynamics} and ``oscillating coherences'' may ensue  despite the underlying dynamical generator being time-invariant \cite{PhysRevResearch.5.L012003,buvca2019non,Zhao2025} -- allowing for time-translation symmetry breaking and the emergence of so-called boundary time crystals \cite{Iemini2018,Yang2025}.

From a system-theoretic standpoint, the challenge of obtaining dynamical models of lower complexity falls under the broad framework of {\em model reduction} (MR) techniques \cite{antoulas}. {\em Exact} MR, in particular, aims to construct models that reproduce the evolution of selected quantities of interest without error, for arbitrary times. A systematic approach to achieve exact MR for quantum Markovian dynamics has been recently developed in a series of papers \cite{grigoletto2023modelreductionquantumsystems, grigoletto2024exactmodelreductioncontinuoustime,letter2024, grigoletto2025quantummodelreductioncontinuoustime},
by combining Krylov operator spaces and algebraic methods. Notably, for continuous-time dynamics, the exact reduced model (if one exists) can be guaranteed to {\em also} be in Lindblad form. In general, however, exact reproduction of the target evolution may be too strong a requirement, with the above operator-algebraic approach resulting in no reduction -- especially for models that are already approximate as they are inferred from noisy tomographic data or imperfect system characterization \cite{Havel,Dobrynin_2025}. {\em Approximate} MR is then often pursued by leveraging time-scale separation ideas, with {\em adiabatic elimination} (AE) being a prominent example for obtaining a linear evolution equation as a series expansion on a low-dimensional operator subspace \cite{azouit2017towards,forni2018adiabatic}. Yet, a recurring problem therein is that approximate reductions may cease to represent a valid Markovian quantum dynamics \cite{riva2024cdc, EssigAllOrder, PhysRevA.109.062206}. This may happen in two conceptually distinct ways. First, the reduced evolution may fail to be described by a valid, completely-positive and trace-preserving (CPTP) super-operator. Second, even when a reduced Lindblad master equation is formally established, and the CPTP property is thus ensured, the construction may implicitly replace density operators by elements of an operator subspace that does not carry an algebraic structure -- preventing their interpretation as valid quantum states. For both quantum-information and device-level interpretations, this distinction matters: A useful reduced model should also point to a physically admissible interpretation of \emph{what the reduced system is}. In finite dimension, this is naturally expressed by requiring the reduced state space to arise from an associative $^*$-algebra (a direct sum of full matrix algebras), so that the reduced states are valid density operators and their dynamics acts on an identifiable effective quantum system \cite{accardi1982quantum}.

In this paper, we show that the algebraic perspective of {exact} MR can be used to obtain approximate MR procedures that remain physically consistent by construction, and  furthermore have a direct connection with AE. More specifically, our main contribution is twofold:

(i) We develop an \emph{asymptotically exact} reduction based on the center (peripheral) manifold of a Lindbladian which, as noted above, governs the long-time behavior of the dynamics. Since such a set has the structure of an associative subalgebra, one can obtain a reduced Markovian dynamics, with the additional property that the reduced long-time evolution is unitary. The associated reduced description is asymptotically exact, in the sense that any initial error outside the center manifold decays exponentially with a rate controlled by the spectral gap of the Lindblad generator, and is the minimal one with such a property. Notably, similar tools are used in classical dynamical-system theory to reduce non-linear and infinite-dimensional systems, see e.g. \cite{kato2013perturbation, roberts1989appropriate, roberts2015macroscale}.

(ii) We consider the case where the evolution of interest may not have a center manifold suitable for the desired reduction, or appears too complicated to perform such reduction directly, but can nonetheless be seen as an (analytic) perturbation of a Lindblad generator we may reduce. In this scenario,  we propose to construct a 
\emph{CPTP-preserving} perturbative reduction. The key choice is to keep the reduction maps fixed at the unperturbed value, so that the reduced state space remains an algebra even when the underlying generator is perturbed. The resulting reduced generator is again Lindbladian for all perturbation strengths for which the original generator is such. As a tradeoff, the reduced system only approximately matches the true system dynamics, and we provide explicit finite-time error bounds that quantify the leakage from the unperturbed center sector. 

Our first contribution hinges on the computation of the center manifold and its structure and, in practice, is typically applicable to idealized systems, which feature manifest exact symmetries or result from specific open-system engineering. Our second contribution allows a generalization to all systems which can be viewed as a perturbation of such a situation. This also encompasses non-idealized settings, including systems with meta-stable oscillations where the actual center manifold is trivial (namely, a single fixed point), as well as systems where an exact computation of the center manifold appears cumbersome.

In establishing the above results, we also gain important insight into the relationships and differences between algebraic and AE approaches. On the one hand, one may observe that, under perturbations, purely imaginary spectral components typically acquire small negative real parts, producing a slow subspace that is invariant, but need no longer carry a $^*$-algebra structure. In that sense, our analysis makes it clear how the AE methods tracking of the ``slow'' space may have better long-time accuracy, at the cost of losing complete positivity or the direct interpretation of the reduced variables as a quantum subsystem. On the other hand, we also detail how, at first order, AE possesses a genuine gauge freedom: different choices lead to different effective generators on the reduced coordinates, with our perturbative CPTP-preserving reduction coinciding with first-order AE for a {\em specific} gauge choice. In doing so, we uncover an important (and hitherto unappreciated) structural role for such a gauge freedom, while additionally providing general conditions under which the first-order AE generator can be chosen to be of Lindblad form. Numerically, we demonstrate how arbitrary, randomized gauge choices perform poorly, whereas enforcing complete positivity improves stability and  long-time performance. Overall, our results delineate a practical tradeoff between accuracy at long times, guaranteed physicality of the reduced dynamics, and algebraic interpretability of the reduced state space.

The content is organized as follows. Section~2 introduces the relevant class of dynamical models and formally defines their center manifold, along with the associated CP projection. Section 3 shows how the existing framework for exact MR may be extended to yield an asymptotically exact MR onto the center manifold, and illustrates its usefulness on a dissipative XXZ spin chain exhibiting persistent oscillations. After reviewing the essential AE formalism, Section 4 shows how, within the AE framework, a reduced generator can be constructed at first order, that retains the CPTP property. Finite-time error bounds are derived, and the algebraic and AE-based MR approaches are quantitatively compared in illustrative perturbed spin-chain models. Three appendixes provide additional background material on both algebraic and adiabatic methods, complete proofs of the general claims and supporting derivation for the dissipative spin-chain example.

\section{Quantum dynamical semigroups and center manifolds}
\smallskip

\subsection{Quantum dynamical semigroups} 
\smallskip

In this paper, we focus on finite-dimensional quantum systems described on a Hilbert space $\Hc\simeq\Cb^n$. The set of linear (bounded) operators acting on $\Hc$ is denoted by $\Bf(\Hc) \simeq \Cb^{n\times n}$, and the set of density operators by $\Df(\Hc)\equiv\{\rho\in\Bf(\Hc)\,|\,\rho=\rho^\dag\geq0,\,\tr(\rho)=1\}$. We will use the notations $[\cdot,\cdot]$ and $\{\cdot,\cdot\}$ to denote the commutator and anti-commutator between two operators, respectively. 

For the class of systems we consider, the evolution of the state $\rho(t)$ in the Schr\"odinger picture is described by a linear, time-invariant equation of the form 
\begin{equation}
    \dot{\rho}(t) = \Lc[\rho(t)] ,
   \label{ME}
\end{equation}
where $\Lc$ is the generator of a valid, completely-positive and trace-preserving (CPTP) {\em quantum dynamical semigroup} \cite{alicki-lendi}. It is well-known \cite{lindblad1976generators,Gorini:1975nb} that a generator $\Lc$ of a (time-invariant) quantum dynamical semigroup, $\{e^{\Lc t}\}_{t\geq0}$, can be expressed in a canonical Lindblad (diagonal) or, if a fixed operator basis is chosen, a Gorini-Kossakowski-Sudarshan-Lindblad (GKLS) form. 
Explicitly, in units where $\hbar=1$, we may write 
\begin{equation*} 
\Lc(\rho) = -i[H,\rho] + \sum_k \Dc_{L_k}(\rho), 
\end{equation*}
where $H=H^\dag\in\Bf(\Hc)$ describes the Hamiltonian (coherent) contribution to the dynamics, $L_k\in\Bf(\Hc)$ are so-called Lindblad (or noise) operators, and the dissipator \(\Dc_{L}(\rho) \equiv L\rho L^\dag - \frac{1}{2}\{L^\dag L,\rho\}\). While the representation of the generator $\Lc$ in terms of $H, \{L_k\}$ is not unique \cite{lindblad1976generators, johnson2015general}, specific forms are often suggested by physical considerations. As mentioned, a master equation of the above form can be taken to describe a wide class of open quantum systems undergoing Markovian dynamics.

\subsection{Center manifold}
\smallskip

At the heart of this work is the notion of the dynamical center manifold. By definition, this corresponds to the eigenspace of $\Lc$ whose eigenvalues have {\em zero} real part. It thus consists of the variables which persist in the long-time dynamics, after either settling down to a stationary state or converging towards a state that supports ``oscillating coherences''  \cite{PhysRevA.89.022118,buvca2019non}. More formally:
\begin{definition}
    Given a Lindblad generator $\Lc$, the {\em center manifold} is defined as 
    \[\Cs \equiv \Span\{X\in\Bf(\Hc)\, | \, \exists\omega\in\Rb\text{ s.t. } \Lc(X) = i\omega X\,\}.\]
\end{definition}
Equivalent (or related) notions of limit cycle sets, peripheral spectral eigen-operators, or rotating points have been discussed in the literature  \cite{wolf2010inverseeigenvalueproblemquantum,wolf2012quantum}, along with their connections to decoherence-free subspaces and information-preserving structures \cite{Schirmer2010,blume2010information}. Note that the center manifold $\Cs$ always includes the set of {\em fixed (stationary, or steady)} operators, 
$$\ker(\Lc) \equiv \{X\in\Bf(\Hc)|\, \Lc(X)=0\} \subseteq \Cs,$$ 
which, in the case of finite-dimensional Lindblad generators as we consider, is always non-empty. Since any such Lindblad generator $\Lc$ is  stable \cite{wolf2012quantum}, $\Cs$ is attractive and any trajectory converges exponentially to it. Let the {\em spectral gap} $\Delta_\Lc$ of $\Lc$ be defined as the absolute value of the real part of the slowest, non-purely-imaginary eigenvalue \cite{PhysRevA.111.052206}\footnote{Sometimes the spectral gap is defined in the literature to quantify the convergence to the steady-state manifold, in which case the maximum is taken over all nonzero eigenvalues; an ``oscillating gap'' has also been introduced \cite{PhysRevB.110.104303}, where the maximum is taken over all non-purely-real eigenvalues.}, that is, \[ \Delta_\Lc \equiv \min_{\substack{\lambda\in{\rm sp}(\Lc)\\ \lambda\notin i\Rb}}|
 \Re(\lambda)|.  \]
While the spectral gap determines the asymptotic convergence rate of arbitrary trajectories to $\Cs$, 
the following bound holds for arbitrary times:
\begin{lemma} [Exponential convergence to $\Cs$]
\label{lem:expo_convergence}
Let $\Lc$ be a Lindblad generator with spectral gap $\Delta_\Lc$
and let $\Pc$ be a projector onto $\Cs$, i.e., $\Pc^2=\Pc$ and ${\rm Im}\,\Pc = \Cs$. Then, for all $\rho_0\in\Df(\Hc)$, and for all $\delta>0$, there exists $\Gamma>0$ such that
\begin{equation}
    \norm{e^{\Lc t} \rho_0 - \Pc(e^{\Lc t} \rho_0) } \leq \Gamma e^{-(\Delta_\Lc-\delta) t}, \quad \forall t\geq0,
    \label{eq:exponential_convergence}
\end{equation}
where $\norm{\cdot}$ is an arbitrary operator norm. 
\end{lemma}

\noindent 
A proof is included in Appendix \ref{app:proofs}. The result is a direct consequence of the fact that finite-dimensional quantum dynamical semigroups are stable, hence $\Lc$ only admits trivial (one-dimensional) Jordan blocks associated to purely imaginary eigenvalues, see e.g., \cite[Prop 6.2]{wolf2012quantum}.

In the following, we use the term super-operator to mean any linear map on $\Bf(\Hc)$, and call any idempotent $\Pc:\Bf(\Hc)\to\Bf(\Hc)$, such that $\Pc^2 \equiv \Pc\circ\Pc = \Pc$ a \textit{projector}. Furthermore, we say that $\Pc$ is \textit{orthogonal} if it is self-adjoint with respect to the Hilbert-Schmidt inner product i.e., $\Pc=\Pc^\dag$. Given another super-operator $\Bc$, we say that $\Pc$ is a \textit{$\Bc$-spectral projector} if it commutes with said super-operator, i.e., $[\Pc,\Bc]=0$.

As remarked in the Introduction, characterizing non-stationary long-time evolution is highly relevant in the context of non-equilibrium many-body dynamics, but also important in scenarios where fast unitary dynamics is present, which one wants to avoid treating as time-dependencies in a rotating frame \cite{riva2024cdc}. Mathematically, a key property of $\Cs$, which we will use in the rest of the paper, is the fact that the $\Lc$-spectral projection onto $\Cs$ is CP. What follows builds on known results, namely, the structure of the fixed points of CPTP maps and Lindblad generators, see e.g., \cite{lindblad1999general,johnson2015general}, and the structure of the eigenspaces associated to limit cycles and peripheral spectra of CPTP maps \cite{wolf2010inverseeigenvalueproblemquantum}. In particular, this structure is discussed in detail in 
\cite[Sec. 2.2]{johnson2015general} for both the case where a full-rank steady state exists and the general case. Adapting \cite[Theorem 8]{wolf2010inverseeigenvalueproblemquantum} to the Lindblad setting, we have the following result:

\begin{proposition}[CPTP projection onto $\Cs$]
\label{prop:spectral_projection}
    Let $\Lc$ be a 
    Lindblad generator, and let $\Cs$ be its center manifold. Then there exists a CPTP, not necessarily orthogonal (w.r.t. $\inner{\cdot}{\cdot}_{HS}$), $\Lc$-spectral projector $\Pc$ onto $\Cs$, i.e.: ${\rm Im}\,\Pc = \Cs$, $\Pc^2 = \Pc$, $[\Pc,\Lc] = 0$ and $\Pc$ is CPTP. 
\end{proposition}

\noindent 
A proof is again included in Appendix \ref{app:proofs}.

\section{Asymptotically Exact Model Reduction}
\smallskip

We now show how restricting the dynamics to the center manifold yields a unitary evolution on an operator subalgebra that  simulates the original one, with errors that decay exponentially fast over time. The method is of direct interest for open dynamics that include non-trivial sets of fixed points and/or oscillatory behaviors which survive in the long-time limit. This is the case, for example, in models where the Hamiltonian commutes with the dissipator, or when information-preserving structures exist (e.g., decoherence-free subspaces, or unitarily evolving subsystems \cite{blume2010information}). A detailed discussion of a physically relevant example immediately follows the general results, in Section \ref{sec:xxz_time_crystal}. In addition, this approach allows us to present all the key ingredients needed in developing the results based on perturbation theory in Section \ref{sec:PT}. 

\subsection{Reduction on the center manifold}
\label{sub: RCM}
\smallskip

We first show that the center manifold $\Cs$ of a Lindblad generator $\Lc$ has an algebraic structure. Proposition \ref{prop:spectral_projection} above establishes that there exists a CPTP projection onto $\Cs,$ and the structure of the fixed points of CPTP maps is well-known (see Appendix \ref{app:aymptotic_model_reduction}).

It follows that there must exist a decomposition of the Hilbert space $\Hc = \bigoplus_k (\Hc_{F,k}\otimes\Hc_{G,k})\oplus\Hc_R$, some unitary $U\in\Bf(\Hc)$ and a set of density operators $\tau_k\in\Bf(\Hc_{G,k})$ such that
\begin{equation}
    \Cs = U \bigg[\bigoplus_k \big(\Bf(\Hc_{F,k})\otimes\tau_k \big) \oplus \zero_R \bigg] U^\dag.
\end{equation}
The set can be interpreted as a direct sum of independent and smaller ``full'' algebras $\Bf(\Hc_{F,k})$, each of them repeated a number of times corresponding to the dimension of $\Hc_{G,k}$ and weighted by the eigenvalues of $\tau_k.$ A standard subalgebra would require all $\tau_k$ to be uniformly mixed states; the more general sets we consider are sometimes referred to as {\em distorted algebras}. In what follows, we assume that $\Cs$ has full support, namely, $\dim(\Hc_{R})=0$; if this is not the case, then one can always first reduce onto $\supp(\Cs),$ the latter being an asymptotically attracting set for the dynamics.

Based on the above structural characterization, one can then observe that $\Cs$ is isomorphic to a {\em smaller} representation, which we denote $\check{\Cs}$, where the weighted copies are removed, i.e., $\check{\Cs} \equiv \bigoplus_k \Bf(\Hc_{F,k})\subseteq\Bf(\check{\Hc})$, with $\check{\Hc}=\bigoplus_k \Hc_{F_k}$. This allows us to find a factorization of the CPTP projector $\Pc$ in terms of two maps \cite{grigoletto2023modelreductionquantumsystems}, 
$\Rc:\Bf(\Hc)\to\Bf(\check{\Hc})$, ${\rm Im}\Rc=\check{\Cs}$ and $\Jc:\check{\Cs}\to\Bf(\Hc)$, ${\rm Im}\Jc={\Cs}$, such that $\Rc\Jc = \Ic_{\check{\As}}$, and $\Pc = \Jc\Rc.$
Furthermore, one can observe that these two factors can be chosen to be CPTP, via the explicit forms:
\begin{equation}
\label{eqn:reduction}
\Rc(X)=\bigoplus_k \tr_{\Hc_{G,k}}(W_k X W_k^\dag)=\bigoplus_k X_{F,k}=\check X , \quad \forall X\in\Bf(\Hc), 
\end{equation}
and for all $\check{X}\in\check{\Cs}$, $\check{X}=\bigoplus_k X_{F,k}$,
\begin{equation}
\label{eqn:injection}
\Jc(\check X)=U\bigg(\bigoplus_k X_{F,k} \otimes\tau_{k}\bigg)U^\dag.
\end{equation}
We name $\Rc$ the \textit{reduction super-operator} and $\Jc$ the \textit{injection super-operator}. The reduction and injection super-operators naturally allow us to compute the reduced dynamics $\check{\Lc} \equiv \Rc\Lc\Jc$ restricted to the algebra $\check{\Cs}$. Importantly, we can prove that the reduced generator $\check{\Lc}$ is a valid Lindblad generator (see Appendix \ref{app:aymptotic_model_reduction}, Thm.\,\ref{thm:Lindblad reduction}). 
Leveraging these properties of the center manifold $\Cs$, we then have the following result:

\begin{theorem}[Asymptotically exact reduction]
\label{thm:asymptotic_reduction}
Let  $\Cs$ be a center manifold of a Lindblad generator $\Lc$. Let then $\Pc$ be the CPTP projection onto $\Cs$, and $\Rc$ and $\Jc$ be the two CPTP factors $\Pc = \Jc\Rc$ as defined in Eqs.\,\eqref{eqn:reduction} and \eqref{eqn:injection}. Define $\check{\Lc} \equiv \Rc\Lc\Jc$. Then the following hold:
    \begin{enumerate}
        \item For all $\rho\in\Cs$ and all $t\geq 0$, we have \[e^{\Lc t}(\rho) = \Jc \,e^{\check{\Lc} t}\, \Rc(\rho).\]
        \item For all $\rho\in\Df(\Hc)$, we have \[\lim_{t\to+\infty} \big[ e^{\Lc t}(\rho) - \Jc e^{\check{\Lc} t} \Rc(\rho)\big] = 0 .\] 
        \item For all $\rho\in\Df(\Hc)$ and for all $\delta>0$, there exists a constant $\Gamma>0$ such that
        \[
        \big| \big|{e^{\Lc t}(\rho) - \Jc e^{\check{\Lc}t}\Rc(\rho)}\big|\big| \leq \Gamma e^{-(\Delta_\Lc-\delta) t},\quad \forall t\geq0 , \]
        where $\Delta_\Lc\in\Rb$ is the spectral gap of $\Lc$.
        \item The restricted dynamics on the center manifold is Lindbladian and unitary, of the form \(\check{\Lc}(\check{\rho}) = -i[\check{H},\check{\rho}]\), where $\check{H} = \bigoplus_k H_{F,k}$ and $H_{F,k}\in\Bf(\Hc_{F,k})$ for all $k$.
    \end{enumerate}
\end{theorem}

\noindent
The proof is provided in Appendix \ref{app:proofs}. Notice that, differently from \cite[Theorem 8, point 2]{wolf2010inverseeigenvalueproblemquantum}, in the continuous-time case there can be {\em no} permutations of blocks of the same dimension. Therefore, each block of the decomposition of $\Cs$ is invariant and evolves independently of the others. 

It is worth remarking that the reduced generator $\check{\Lc}$ is the {\em minimal} model (in the sense of the smallest dimension) that is capable of exactly reproducing the asymptotic behavior (that is, at arbitrarily long times) of the reference system: for a reduced model to correctly reproduce the asymptotic dynamics for {\em all} initial conditions, the associated subspace must necessarily contain the entire center manifold, and the minimal projection onto such subspace is precisely $\Pc$. It is also interesting to note that, differently from discrete-time evolutions, in continuous time coherences can be preserved in the long time regime only in the presence of at least two steady states. In other words, in the case of a unique steady state, the center manifold coincides with the null space of the generator, see e.g. \cite[Theorem 4]{buvca2022algebraic}.

The results we obtained thus far may be aptly illustrated in the context of a paradigmatic many-body dynamical model, as we turn to next.

\subsection{Case study: Non-stationary coherent long-time dynamics in a dissipative XXZ chain} 
\smallskip
\label{sec:xxz_time_crystal}

Consider a system composed of $N$ spin-$\um$ degrees of freedom (qubits) in one dimension, whereby $\Hc = \otimes_{i=1}^{N} \Hc_i,$ $\Hc_i\simeq\Cb^2,$ and $\Bf(\Hc) \simeq (\Cb^4)^{\otimes N}.$ As usual, let $\sigma_q$, with $q\in\{0,x,y,z\}$, denote the Pauli spin-1/2 matrices and the identity $\sigma_0\equiv \one_2$. We will use the shorthand notation $\sigma_q^{(k)}\in\Bf(\Hc)$ to denote their multi-spin extension acting non-trivially only on the $k$-th spin, i.e., $\sigma_q^{(k)} \equiv \otimes_{i=1}^{k-1} \sigma_0 \otimes \sigma_q \otimes_{i=1}^{N-k}\sigma_0$. Similarly, we define local ladder (raising and lowering) operators acting non-trivially on spin $k$ as $\sigma_\pm ^{(k)} \equiv \frac{1}{2}(\sigma_x^{(k)}\pm i\sigma_y^{(k)})$. We further assume periodic boundary conditions, hence $\sigma_q^{(N+1)}=\sigma_q^{(1)}$, $\forall q$. 

We take the spins to have a common natural frequency $\omega$ and interact via an XXZ Hamiltonian,  
\begin{align}
    H = \sum_{j=1}^N \bigg [ \frac{\omega}{2} \,\sigma_z^{(j)} + 
    \frac{A_{xy}}{2}\big(\sigma_x^{(j)}\sigma_x^{(j+1)} + \sigma_y^{(j)}\sigma_y^{(j+1)}\big)  
    + A_z \,\sigma_z^{(j)}\sigma_z^{(j+1)}    
    \bigg ], 
    \label{eq:xxz_ham}
\end{align}
where for concreteness we assume an isotropic in-plane ferromagnetic coupling, $A_{xy}\geq 0$, and the Ising coupling strength $A_z \in {\mathbb R}$. Since, for spin-$1/2$ systems with nearest-neighbor couplings in one dimension, no purely imaginary eigenvalues have been shown to emerge with only on-site non-Hermitian Lindblad operators \cite{buvca2022algebraic}, either spin-$1$ XXZ chains with local (onsite) dissipation or collective spin-$1/2$ models have been considered in the literature, in connection with notions of quantum synchronization and the emergence of time-crystalline behavior \cite{Iemini2018,Tindall2020,Zhao2025,mondkar2026}. Here, we instead introduce dissipation through two-site operators that describe nearest-neighbor incoherent hopping, akin to a quantum version of an ``asymmetric simple exclusion process'' \cite{Derrida,PhysRevE.102.062210}\footnote{Notably, in the purely dissipative limit ($H\equiv 0$), the model has been shown to be integrable \cite{PhysRevE.102.062210}, in the sense that the operator space  ``fragments'' into exponentially many invariant subspaces, for which the projected generator maps to an integrable XXZ Hamiltonian. For $H\ne 0$, integrability is preserved as long as $A_{xy} =0$.}:
\begin{equation}
    L_j = \sqrt{\gamma}\,\sigma_+^{(j)}\sigma_-^{(j+1)},\quad  \forall j=1,\dots,N,\quad \gamma>0.
    \label{eq:xxz_jump}
\end{equation}

Let $J_u \equiv \um\sum_{j=1}^N \sigma_u^{(j)}$ denote total angular momentum operators. Then $J_z$ is a conserved quantity for the model, i.e., it is invariant under the Heisenberg-picture generator associated to Eq.\,\eqref{ME}, $\Lc^\dag(J_z) = 0$, and $e^{-i\varphi J_z}$ forms a strong $U(1)$ symmetry. A proof of both these facts is given in Appendix \ref{app:xxz_extra}. Because $J_z$ is the generator of a strong symmetry, we can decompose $\Hc$ into eigenspaces of $J_z$, with each one of them supporting at least one steady state \cite[Theorem A.1]{Buca_2012}. Let $m$ denote the eigenvalues of $J_z$, with $m\in \{ -\frac{N}{2},\dots,\frac{N}{2}\}$, and let $\{ \ket{s_m} \}$ be a corresponding complete set of orthonormal eigenvectors, satisfying $J_z\ket{s_m} = m \ket{s_m}$. It is well known \cite{franchini2017introduction} that {each eigenvector $\ket{s_m}$ is a linear combination of bitstrings} with exactly $\frac{N}{2}-m$ spin down ($\ket{1}$) and $\frac{N}{2}+m$ spin up ($\ket{0}$, where we used the convention $\sigma_z\ket{0}=\ket{0}$). We thus have $\Hc = \bigoplus_m \Hc_m$, where $\Hc_m = \Span\{\ket{s_m}\}$. Explicitly, we have:
\begin{align*}
    \Hc_{\frac{N}{2}} &= \Span\{\ket{00\dots0}\} ,\\
    \Hc_{\frac{N}{2}-1} &= \Span\{\ket{10\dots0} , \ket{01\dots0} ,\dots, \ket{00\dots1}\},\\
    &\,\,\,\vdots\\
    \Hc_{-\frac{N}{2}} &= \Span\{\ket{11\dots1}\}.
\end{align*}
Furthermore,  the dimension $d_m \equiv \dim(\Hc_m) = \binom{N}{\frac{N}{2}-m}.$

Clearly, the two states $\bar{\rho}_{\frac{N}{2}}\equiv\ketbra{00\dots0}{00\dots0}$ and $\bar{\rho}_{-\frac{N}{2}}\equiv\ketbra{11\dots1}{11\dots1}$ are steady states, since they are the only states supported on their respective Hilbert subspaces. Let us then define the set of states 
\begin{equation}
    \bar{\rho}_m \equiv \sum_{s_m} \frac{1}{d_m} \ketbra{s_m}{s_m}, \quad m \in \big\{ \!-\tfrac{N}{2},\dots,\tfrac{N}{2}\big\}.
    \label{eq:steady_states}
\end{equation}
As it turns out, these $N+1$ states form the steady-state manifold, i.e., $\ker\Lc = \Span\{\bar{\rho}_m, \, m\in\Mb\}$. The fact that $\Lc(\bar{\rho}_m) = 0$ for all $m$ was noted for a similar spin-$1$ model \cite{Tindall2020,Zhao2025} and is explicitly verified in our case in Appendix \ref{app:xxz_extra}. We also verify there [Theorem \ref{thm:uniqueness_steady_states}] that, for almost all values of $A_{xy}$, $\ker\Lc$ contains no other operators.

Let $\ket{0_L} \equiv \ket{00\dots0}$ and $\ket{1_L} \equiv \ket{11\dots1}$. One can observe that $\Lc$ has two purely imaginary eigenvalues: 
$\Lc(\ketbra{0_L}{1_L}) = -i\omega_0 \ketbra{0_L}{1_L}$ and $\Lc(\ketbra{1_L}{0_L}) = i\omega_0 \ketbra{1_L}{0_L}$, with $\omega_0 \equiv \omega N$ (see Appendix \ref{app:xxz_extra} for a proof). This lets us conclude that the center manifold is given by 
\begin{align}
    \Cs &= \Span\big\{\ketbra{0_L}{1_L}, \ketbra{1_L}{0_L}\big\} \oplus \ker\Lc  
    \label{eq:center_manifold}
    \simeq 
    \Cb^{2\times 2} \oplus \bigoplus_{m \in  \mathbb{M}} 
    \Cb \one_{d_m}, \quad 
    \mathbb{M} \equiv \big\{ \! -\tfrac{N}{2}+1, \ldots, \tfrac{N}{2}-1 \big\},  
\end{align}
where the full $\Cb^{2\times 2}$ block is associated to the subalgebra $\As = \Span\{\ketbra{j_L}{k_L}, j,k=0,1\}$, and may be thought of as encoding a logical qubit. Again, the fact that $\Cs$ contains no other operators is proven in Appendix \ref{app:xxz_extra} for $A_{xy}=0$, and tested numerically for other values of $A_{xy}$ up to $N=6$. $\Cs$ is then isomorphic to 
\[
\check{\Cs} = \Cb^{2\times 2} \oplus \bigoplus_{m\in \mathbb{M}} 
\Cb \subset \Cb^{N+1\times N+1}.
\]
Let $U\in\Bf(\Hc)$ denote the permutation matrix that has $\ket{0_L},\ket{1_L}$ as the first two columns, followed by the columns $\ket{s_m}$ grouped together according to the value of $m$. The CPTP projector $\Pc$ onto $\Cs$ is then given by 
\[
\Pc(\rho) = W_0  \rho W_0^\dag \bigoplus_{m \in \mathbb{M}}
\tr[W_m \rho W_m^\dag] \,\frac{\one_{d_m}}{d_m} , 
\]
where $W_0 \equiv \big[\begin{array}{c|c} \one_2  & 0_{2\times 2^{N-1}} \end{array}\big] U^\dag \in\Cb^{2\times 2^N}$ and $W_m \equiv \big[\begin{array}{c|c|c} 0 & \one_{d_m}  & 0 \end{array}\big] U^\dag \in\Cb^{d_m\times 2^{N-1}}$. The respective CPTP reduction and injection maps are, in turn, given by
\begin{align*}
    \Rc(\rho) &= W_0  \rho W_0^\dag \bigoplus_{m \in \mathbb{M}}
     \tr[W_m \rho W_m^\dag] ,&
    \Jc(\check{\rho}) &= p_{\frac{N}{2}}\sigma \bigoplus_{m \in \mathbb{M}}
      p_m \frac{\one_{d_m}}{d_m} ,
\end{align*}
where the time-evolved reduced state $\check{\rho}$ is
\[ \check{\rho}(t)
= p_{\frac{N}{2}} \sigma(t) \oplus \bigoplus_{m \in \mathbb{M}}
p_m , \]
with $p_m\geq 0$, $p_{\frac{N}{2}}+\sum_{m \in \mathbb{M}}
p_m = 1,$ and $\sigma(t) \in\Df(\Cb^2)$, $p_{\frac{N}{2}} \equiv \tr[W_0  \rho W_0^\dag]$.
With these, we can obtain the reduced Hamiltonian $\check{H} = \check{H}_{\frac{N}{2}} \oplus \bigoplus_{m\in \mathbb{M}} 
0 =\check{H}_{\frac{N}{2}} \oplus 0_{(N-1)\times (N-1)}$, where $\check{H}_{\frac{N}{2}} = \frac{\omega_0}{2} \sigma_z. $

\begin{figure*}[t]
    \centering
    \begin{subfigure}[c]{0.49\textwidth}
        \includegraphics[width=\linewidth]{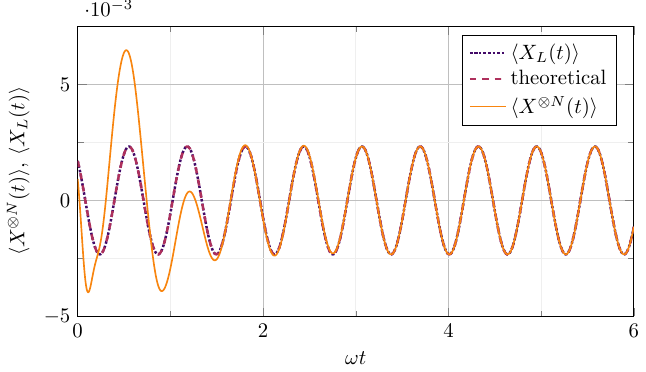}
    \end{subfigure}
    \begin{subfigure}[c]{0.49\textwidth}
    \includegraphics[width=\linewidth]{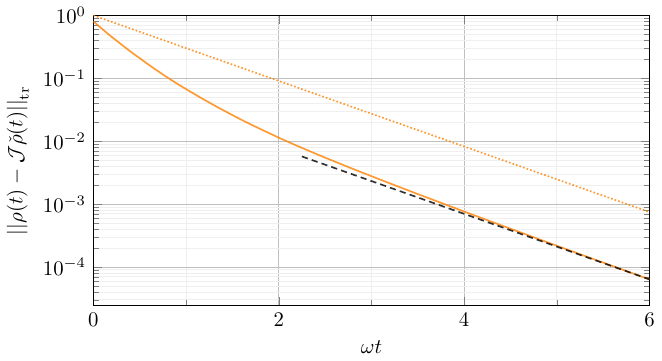}
    \end{subfigure}
    \caption{{\small 
Left: Expectation values of the full and logical observables $X^{\otimes N}$ and $X_L$ versus time, for a system of $N=5$ spins. The dashed line represents the theoretical solution of Eq.\,\eqref{eq:theo_exp_val}.  Right: Exponential convergence of the trace-norm distance between $\rho(t)$ and $\Jc(\check{\rho}(t))$. The continuous line represents the norm computed numerically in a simulation, the dotted line represents the exponential bound given by $\Gamma e^{-\Delta_\Lc t}$, with $\Gamma =1$ and $\Delta_\Lc =0.72$ and the dashed line shows how $\Delta_\Lc$ correctly sets the asymptotic convergence rate. Here the initial condition was chosen randomly while the remaining parameters are $\omega= 1.2$, $A_{xy}=2$, $A_x = 4.6$, and $\gamma=1.2$.} 
}
    \label{fig:red_xxz_crystal}
    \label{fig:kl-convergence}
\end{figure*}

Let $X^{\otimes N} \equiv\prod_{j=1}^{N}\sigma_x^{(j)}$ denote the $N$-body $X$ observable. By writing $X^{\otimes N}$ in the basis defined by $U$, one can observe that $U X^{\otimes N} U^\dag$ has support over the $2\times2$-block, i.e., $W_0 X^{\otimes N} W_0^\dag = \sigma_x$, while it has no support over the remaining diagonal blocks. This implies that the expectation of $X^{\otimes N}$ will converge, as $t\to+\infty$, to the expectation of $ X_L \equiv\ketbra{0_L}{1_L} + \ketbra{1_L}{0_L}$ (the one with support only on the $2\times2$-block).  By also letting $Z_L  \equiv \ketbra{0_L}{0_L} - \ketbra{1_L}{1_L}$ and $Y_L \equiv iX_LZ_L$, we can find an analytical solution for $\expect{X_L(t)}$ for all $t\geq0$, namely, 
\begin{align}
    \expect{X_L(t)} 
    = a_x\cos\left(\omega_0 t\right) - a_y\sin\left(\omega_0 t\right), \label{eq:theo_exp_val}
\end{align}
where $a_x \equiv \tr(X_L \rho_0)$ and $a_y \equiv \tr(Y_L \rho_0)$. Thus, $\lim_{t\to+\infty} [\expect{X^{\otimes N}(t)} -\expect{X_L(t)}] = 0$, meaning that the spin subsystems exhibit {\em stable quantum synchronization} with respect to the observable $X^{\otimes N}$, independently of the initial condition  \cite{buvca2019non,buvca2022algebraic}. In fact, a behavior similar to Eq.\,\eqref{eq:theo_exp_val} is also found for $Y_L(t)$, with $\lim_{t\to+\infty} [\expect{Y^{\otimes N}(t)} -\expect{Y_L(t)}] = 0$, whereas $Z^{\otimes N}$ is (trivially) stationary. 

It is interesting to observe that the reduced Hamiltonian $\check{H}_{\frac{N}{2}}$, which governs the dynamics in the asymptotic time limit, depends only on the energy splitting $\omega$ of the spins and their number $N$, but not on the interaction parameters $A_{xy}$ and $A_z$. This observation provides an example of how MR may be useful in revealing features of the dynamics that need not be {\em a priori} obvious. In Fig.\,\ref{fig:red_xxz_crystal}, we show numerical simulations for $N=5$ spins, starting from a randomly chosen initial condition. In particular, the left panel explicitly demonstrates the convergence of $\expect{X^{\otimes N}(t)}$ to the predicted asymptotic expectation value $\expect{X_L(t)}$, following a transient. In the right panel, the asymptotic convergence towards the center manifold $\Cs$ is also demonstrated, consistent with Theorem \ref{thm:asymptotic_reduction}.

\section{Parametric Models and Perturbation Theory}
\label{sec:PT}

\smallskip

In realistic systems, reduction to an exact center manifold may not always be the most desirable goal. In fact, only an idealized and highly symmetric system would typically feature a non-trivial and well-identifiable center manifold, whereas realistic systems of interest would only feature meta-stable oscillations (i.e., damped at a very low rate), which we may still want to capture in the reduced model. In other situations, the center manifold may be nontrivial for the actual system, but too complex to determine.
In particular, explicitly characterizing the center manifold $\Cs$  and implementing the CPTP projection $\Pc$ yielding the desired (unitary) reduced dynamics may quickly become analytically and numerically intractable as the system size grows, as it is the case in many-body systems. Different approaches to compute $\Cs$ have been put forward, that allow for the problem to be circumvented in special scenarios: for example, symmetry principles are invoked in \cite{PhysRevResearch.5.L012003} to construct noise operators that yield coherent oscillations in the long-time regime, while necessary and sufficient conditions for purely imaginary eigenvalues to emerge  and  a characterization of $\Cs$ given the Hamiltonian and noise operators are proposed in \cite{buvca2022algebraic}. It is also worth mentioning that, just like progress has been made in determining exact steady-state solutions for various Lindbladian models of interest \cite{ClerkHTR}, similar preliminary results have been obtained for certain oscillatory behaviors \cite{tokieda2026exact}.

In view of the above, here we shift our focus to a {\em perturbative} approach to MR: we assume that the model at hand may be described as a perturbation of another model, where the latter admits an easy to be determined, nontrivial center manifold. In this scenario, we show how an approximate CPTP reduced model can be obtained, whose dimension is the one of the unperturbed center manifold, i.e.~nontrivial and including possible meta-stable oscillations.

\subsection{Perturbative setting}
\smallskip
\label{sub:pert}

To develop a perturbative scheme, we assume that the evolution in Eq.\,\eqref{ME} is governed by a generator $\Lc_\varepsilon$ which is Lindblad for all $\varepsilon$ in a continuous interval $\Sc\subseteq \Rb$ that contains $0$, and depends {\em analytically} upon $\varepsilon$. Thanks to the analyticity assumption, a Taylor expansion around $\varepsilon=0$ is then possible:
\[\Lc_\varepsilon = \sum_{k=0}^{+\infty} \varepsilon^k \Lc_{(k)},\; \] 
where the unperturbed generator $\Lc_{(0)}\equiv \Lc_{\varepsilon=0}$ is assumed to be in Lindblad form and sufficiently simple to study, while the terms $\Lc_{(k)}$ represent the perturbing contributions. For example, one can consider the paradigmatic case where $\Lc$ may be expressed in the form $\Lc\equiv \Lc_0+\Delta\Lc,$ where $\Lc_0$ is Lindblad and has a known, or easy to compute, center manifold $\Cs_0,$ and $\|\Delta\Lc\|=\varepsilon_0\ll \|\Lc_0\|$ for some super-operator norm. Then, one can formally define $\Lc_{(0)}\equiv \Lc_0,$ $\Lc_{(1)}\equiv \Delta\Lc/\varepsilon_0$ and consider the linearly perturbed model $\Lc_\varepsilon=\Lc_{(0)}+\varepsilon\Lc_{(1)},$ which matches the model of interest for $\varepsilon=\varepsilon_0.$

Examples of relevant scenarios that fit this framework include: (i) a (finite-dimensional) system with a {\em weak Hamiltonian  coupling to a structured environment}, whose dynamics are modeled with a combination of Hamiltonian and dissipative generators \cite{Tamascelli_2018}; (ii) {\em approximately preserved quantum codes} in continuous-time quantum-information models \cite{lidar2013quantum}: the dominant dynamics exhibits symmetries and allows for information preservation, but small perturbations cause information leakage and dissipation outside the code (e.g., a multi-qubit system evolving under {\em collective} Markovian decoherence and Hamiltonian terms, which ensures the existence of a decoherence-free subspace or noiseless subsystem \cite{KLV,lidar2013quantum}, subject to small local dissipation effects); (iii) more generally,  {\em symmetry-breaking perturbations}: Hamiltonian and strong dissipative dynamics that admit symmetries and preserved states, together with small perturbations that break such symmetries \cite{gu2024spontaneous}.

Note that, for $\Lc_\varepsilon$ to be a Lindblad generator, a {\em sufficient} condition is to require that all orders of $\Lc_{(k)}$ are Lindblad generators. While this is a common assumption in the adiabatic elimination literature \cite{azouit2017towards,PhysRevA.109.062206}, it is {\em not} necessary: for instance, \(\Dc_{L+\varepsilon K}(\rho)\) is a valid Lindblad generator for all $\varepsilon\in\Rb,K,L$; however, its Taylor expansion takes the form 
\begin{equation}
\Dc_{L+\varepsilon K}(\rho) = \Dc_L(\rho) + \varepsilon\Big(L\rho K^\dag + K \rho L^\dag - \tfrac{1}{2}\{K^\dag L+L^\dag K,\rho\}\Big) +\varepsilon^2 \Dc_K(\rho),\label{eq: LK}
\end{equation}
where the first-order term need not be in Lindblad form\footnote{Taking e.g., $K=\sigma_x$ and $L=\sigma_z$ on a qubit, the first-order term yields the Bloch equations $\dot{x}= 2 z$, $\dot{y}=0$, $\dot{z}=2x$. These {\em cannot} correspond to a valid Lindbladian, since the generator features a positive eigenvalue.}. In what follows, we shall only assume that the full generator $\Lc_\varepsilon$ is a valid Lindblad generator for all $\varepsilon\in\Sc$, while no such requirement is imposed on each of the $\Lc_{(k)}$.

Within the above perturbation setting, we can highlight two different, yet related, center manifolds; specifically:
\begin{itemize}
    \item $\Cs_0$: the center manifold of the unperturbed generator $\Lc_{(0)}$ (with $\check{\Cs}_0$ its isomorphic smaller representation);
    \vspace*{-1mm}
    \item $\Cs_\varepsilon$: the center manifold of the perturbed generator $\Lc_\varepsilon$.
\end{itemize}
For sufficiently small $\varepsilon$, one can prove that $\dim(\Cs_\varepsilon) \leq \dim(\Cs_0)$, see e.g., \cite{kato2013perturbation}.  In what follows, we assume to be able to compute $\Cs_0$.

\begin{remark}[On the robustness of oscillatory modes]
As we already mentioned, the center manifold is not generically robust to arbitrary perturbations. Persistent coherent oscillations, in particular, are associated with purely imaginary eigenvalues of the generator, and generic perturbations will typically move these eigenvalues into the left half-plane, thereby turning such oscillations into exponentially damped modes. Particular cases where oscillations are not disrupted by the perturbations include Hamiltonian perturbations that leave the center manifold (or parts of it) invariant, potentially shifting the oscillation frequencies without destroying the oscillations themselves. By contrast, dissipative perturbations acting non-trivially on the unperturbed center manifold, or Hamiltonian perturbations that couple the center manifold to decaying modes, will induce decoherence and damping. Thus, in general robustness must be analyzed on a case-by-case basis. Relevant results in this sense are presented in \cite{loizeau2026krylovspaceperturbationtheory}. 
\end{remark}

\subsection{Approximate model reduction: Reduced dynamics on the unperturbed center manifold}

\smallskip
\label{sec:algebraic_reduction}

The possibility of obtaining reduced Lindbladians on algebras, that the fundamental Theorem \ref{thm:Lindblad reduction} included in Appendix \ref{app:aymptotic_model_reduction} affords, is a result applicable to a setting far more general than that of Theorem \ref{thm:asymptotic_reduction}: it ensures that restricting \emph{any} Lindblad generator to the center manifold of a \emph{possibly different} Lindblad generator retains the Lindblad property. By taking these Lindblad generators to be $\Lc_\varepsilon$ and $\Lc_0$, respectively, we can thus obtain a first-order approximation of the long-time dynamics associated to $\Lc_\varepsilon$, guaranteed to be in Lindblad form, by projecting it onto the unperturbed center manifold $\Cs_0$. Of course, since $\Cs_0$ does not match an invariant eigenspace of $\Lc_\varepsilon$, in general, Theorem \ref{thm:asymptotic_reduction} does not apply. As a consequence, for $\varepsilon \neq 0$, the reduced model is in general {\em approximate for all times}. The error the reduced model incurs when the initial state belongs to $\Cs_0$ can be bounded in two ways, comparing the evolution either in the full state space or in $\check{\Cs}$. The following two statements formalize these results.

\begin{corollary}[Approximate Lindblad reduction]
\label{cor:projected_perturbed_dynamics}
    Let $\Lc_\varepsilon$ be a Lindblad generator for all $\varepsilon\in\Sc$, with $\Lc_{(0)} = \Lc_{\varepsilon=0}$. Let $\Cs_0$ be the center manifold of $\Lc_{(0)}$, $\Pc_{(0)}$ be the CPTP $\Lc_{(0)}$-spectral projector onto $\Cs_0$ and $\Rc_{(0)}$ and $\Jc_{(0)}$ the CPTP factors of the projector $\Pc_{(0)}$, as defined in Eqs. \eqref{eqn:reduction} and \eqref{eqn:injection}.  
    Then the reduced generator $\check{\Lc}_{\varepsilon} \equiv \Rc_{(0)}\Lc_\varepsilon\Jc_{(0)} = \sum_{k=0}^{+\infty}\varepsilon^k \Rc_{(0)}\Lc_{(k)}\Jc_{(0)}$ is Lindblad and leaves $\check{\Cs}_0$ invariant.
\end{corollary}

\begin{theorem}[Error bounds]
    \label{thm:error_bounds}
    Let $\Lc_\varepsilon$ be a Lindblad generator analytic in $\varepsilon$, with $\Lc_{(0)} = \Lc_{\varepsilon=0}$. Let $\Pc_{(0)}$ be the CPTP $\Lc_{(0)}$-spectral projector onto the unperturbed center manifold $\Cs_0$, and let $\check{\Lc}_\varepsilon \equiv \Rc_{(0)}\Lc_\varepsilon\Jc_{(0)}$. Then, for all $t\geq0$, we have: 
\begin{align}
       \mbox{\it (i)} \;\;\;&  \norm{e^{t\Lc_\varepsilon}\Pc_{(0)} - \Jc_{(0)}e^{ t\check{\Lc}_\varepsilon}\Rc_{(0)}}_{\rm sop}
       \leq  t\, C^2|\varepsilon| \norm{\Lc_{(1)}}_{\rm sop} + O(t\varepsilon^2 + t^2|\varepsilon|),
        \label{eqn:bound_1} \\
        \mbox{\it (ii)} \;\;\;& \norm{\Pc_{(0)}e^{t\Lc_\varepsilon}\Pc_{(0)} - \Jc_{(0)}e^{ t\check{\Lc}_\varepsilon}\Rc_{(0)}}_{\rm sop} \leq \frac{t^2}{2}\,C^3\varepsilon^2 \norm{\Lc_{(1)}}_{\rm sop}^2 
        +  O(t^2|\varepsilon|^3 + t^3 \varepsilon^2) ,
        \label{eqn:bound_2}
    \end{align}
where $\norm{\cdot}_{\rm sop}$ denotes any super-operator norm and we have defined $C\equiv\norm{\Pc_{(0)}}_{\rm sop}$.
\end{theorem}

\noindent 
Corollary \ref{cor:projected_perturbed_dynamics} is a direct consequence of Theorem \ref{thm:Lindblad reduction}. A proof of Theorem \ref{thm:error_bounds} is included in Appendix \ref{app:proofs}. Note that the above bounds are only useful as long as their respective right-hand-side remains small. This is, in particular, the case for times $t$ that are sufficiently small in comparison to the perturbation characteristic time scale $\varepsilon^{-1}$.
Since Eq.\,\eqref{eqn:bound_1} is dominated by how quickly the true evolution may exit the unperturbed center manifold $\Cs_0$, it does no better than integrating the perturbing vector field over time. By following the trajectory projected onto $\Cs_0$, Eq.\,\eqref{eqn:bound_2} provides a bound that scales slower than linearly in $\varepsilon t$.

\begin{remark}[Emerging unitarity]
As a special case of our formalism, we can recover existing results on ``emerging unitarity'' in strongly dissipative dynamics and dissipation-assisted computation over noiseless subsystems, as investigated in \cite{PhysRevLett.113.240406, zanardi2015geometry}. Specifically, linearly perturbed models of the form $\Lc_\varepsilon = \Lc_{(0)} + \varepsilon \Lc_{(1)}$ are considered therein, and the center manifold of the unperturbed generator $\Lc_{(0)}$ is assumed to consist {\em only} of steady states, $\Cs_0\equiv \ker\Lc_{(0)}$, with $\varepsilon \equiv 1/T$, and $T>0$ much longer than the natural relaxation timescale. The ``dissipation-projected dynamics'' is then obtained via projection onto the fixed-points set. Assuming that the perturbative term is purely Hamiltonian, i.e., $\Lc_1(\cdot)=-i[H_1,\cdot]$, it is shown in the above papers that the projected dynamics $\Pc_{(0)}\Lc_{\varepsilon}\Pc_{(0)}$ onto $\ker\Lc_0$ is also purely Hamiltonian. This can be seen as an application of Corollary \ref{cor:projected_perturbed_dynamics} in conjunction with \cite[Corollary 1]{grigoletto2024exactmodelreductioncontinuoustime}. Likewise, the emergence of unitary dynamics over the steady-state manifold is discussed in a scenario where the Lindblad operators are linearly perturbed \cite{zanardi2015geometry,Arenz_2020}, namely, $L\mapsto L+\varepsilon K$. As noted in Eq.\,\eqref{eq: LK}, such a perturbation generates, to order $\varepsilon^1$, a contribution of the form $L\rho K^\dag + K \rho L^\dag - \frac{1}{2}\{K^\dag L + L^\dag K,\rho\}$: when projected onto $\ker\Lc_0$, one may verify that this indeed yields a Hamiltonian generator, as claimed in \cite[Proposition 4]{zanardi2015geometry}. Using Corollary \ref{cor:projected_perturbed_dynamics} and the procedure described in \cite[Appendix]{grigoletto2025quantummodelreductioncontinuoustime}, our formalism also makes it possible to characterize the effects of the perturbation to $O(\varepsilon^2)$, which are encoded in $\Dc_{K}(\rho)$ and reintroduce dissipation in general. 
\end{remark}

\subsection{Approximate model reduction via adiabatic elimination}
\smallskip

An important set of methods that operates under assumptions similar to those of Sec.\,\ref{sub:pert} is AE \cite{azouit2017towards}, which is a specialization to the quantum domain of classical time-scale separation techniques for linear systems, and can be viewed as a direct generalization of Hamiltonian perturbation theory. We first briefly recall the essential ideas, building on results from perturbation theory for linear operators \cite{kato2013perturbation}, as seen through the lenses of our current framework. We then provide a new result on the ``gauge degrees of freedom'' that are intrinsic in the AE formalism \cite{azouit2017towards,riva2024cdc} [Theorem \ref{thm:ad_el_order_1}], highlighting the connection with the approximate Lindblad MR provided in Corollary \ref{cor:projected_perturbed_dynamics}.

\subsubsection{The slow operator subspace.} A small perturbation of the Lindbladian $\mathcal{L}_0$ typically implies a hybridization of other degrees of freedom with the persistent oscillations and steady states of the center manifold $\Cs_0$, inducing their slow damping (a paradigmatic example being in Purcell decay in quantum optics \cite{Houck}. In principle, a block-spectral decomposition should be able to focus on those slowly (or non-) decaying degrees of freedom only. These ideas can be formalized as follows.

Let us denote by $\Ds_\varepsilon$ the {eigenspace of $\Lc_{\varepsilon}$ corresponding to perturbing} the center manifold $\Cs_0$ of $\Lc_{(0)}$. That is, $\Ds_\varepsilon$ is the operator subspace with, say, $\dim(\Ds_\varepsilon) = \dim(\Cs_0) \equiv d_0$,
spanned by the eigenoperators associated to the eigenvalues $\lambda_\varepsilon$, which for $\varepsilon=0$ satisfy $\Re(\lambda_{\varepsilon=0})=0$. For $\varepsilon$ small enough, perturbation theory \cite{kato2013perturbation} implies that $\Cs_\varepsilon\subseteq\Ds_\varepsilon$, with $\Ds_\varepsilon$ corresponding to the slowest eigenvalues of the perturbed generator. Because $\Ds_\varepsilon$ is a sum of eigenspaces, it is $\Lc_\varepsilon$-invariant. This allows one to study the evolution of the system restricted to the {\em perturbed} center manifold $\Ds_\varepsilon$ -- resulting in a reduced model, defined on $\check{\Cs}_0$, that describes the evolution of the slowest components of the dynamics.

Let us denote by $\Pc_\varepsilon$ the $\Lc_\varepsilon$-spectral projector onto $\Ds_\varepsilon$, i.e., ${\rm Im}(\Pc_\varepsilon) = \Ds_\varepsilon$ and $[\Lc_\varepsilon, \Pc_\varepsilon]=0$. Since we are projecting onto the {\em perturbed} space $\Ds_\varepsilon,$ a CPTP $\Pc_{\varepsilon}$ might no longer exist, in general. From \cite{kato2013perturbation}, we know that $\Pc_\varepsilon$ is analytic in $\varepsilon$. Its Taylor expansion reads
\begin{align*}
    \Pc_\varepsilon = \sum_{k=0}^{+\infty} \varepsilon^k \Pc_{(k)}, 
\end{align*}
where $\Pc_{(0)}$ is the zeroth-order term, $\Pc_{(0)} = \Pc_{\varepsilon=0}$, which coincides with the $\Lc_{(0)}$-spectral projector onto $\Cs_0$. Ref.\,\cite{kato2013perturbation} further proves that $\Pc_\varepsilon$ can be obtained from $\Pc_{(0)}$ through an analytic change of basis $\Tc_\varepsilon$, namely:
\begin{align}
    \Pc_\varepsilon = \Tc_\varepsilon \Pc_{(0)} \Tc_\varepsilon^{-1}.
    \label{eq:analitic_change_of_basis}
\end{align}

By Proposition \ref{prop:spectral_projection}, it follows that $\Pc_{(0)}$ is CPTP and, as we have seen above, it can be factorized into $\Pc_{(0)} = \Jc_{(0)}\Rc_{(0)}$, with {$\Rc_{(0)}\Jc_{(0)} = \Ic$, and $\Rc_{(0)}$ and $\Jc_{(0)}$ as in Eqs.\,\eqref{eqn:reduction} and \eqref{eqn:injection}. If we now substitute the factorization of $\Pc_{(0)}$ into Eq.\,\eqref{eq:analitic_change_of_basis}, we obtain
\[
\Pc_\varepsilon = \big(\Tc_\varepsilon \Jc_{(0)} \big) \,\big( \Rc_{(0)} \Tc_\varepsilon^{-1} \big ) \equiv \Jc_\varepsilon \, \Rc_\varepsilon, \] 
a factorization in terms of two functions, $\Rc_\varepsilon:\Bf(\Hc)\to\check{\Cs}_0$ and $\Jc_\varepsilon:\check{\Cs}_0 \to \Bf(\Hc)$, which are analytic in $\varepsilon$ and obey ${\rm Im}\,\Jc_\varepsilon = \Ds_\varepsilon$ and  $\Rc_\varepsilon\Jc_\varepsilon=\Ic_{\check{\Cs}_0}.$ Since the mapping between $\Cs_0$ and $\check{\Cs_0}$ can be defined up to any isomorphism, which reflects an internal ``choice of coordinates'', the two super-operators $\Rc_\varepsilon$ and $\Jc_\varepsilon$ are not unique in general. In the AE context, this isomorphism may itself depend on $\varepsilon$ and it is referred to as a ``gauge degree of freedom'' \cite{azouit2017towards,riva2024cdc}.

To obtain the reduced model on $\Dc_\varepsilon$, let us consider a state $\rho_0\in\Ds_\varepsilon$. Since $\Ds_\varepsilon$ is $\Lc_\varepsilon$-invariant, $\rho(t) = e^{\Lc_\varepsilon t}\rho_0\in\Ds_\varepsilon$ for all $t\geq0$. Then, since $\Pc_\varepsilon$ acts as the identity over $\Ds_\varepsilon$, we have \[\rho(t) = \Pc_\varepsilon e^{\Lc_\varepsilon t} \Pc_\varepsilon \rho_0 = \Jc_\varepsilon e^{\Rc_\varepsilon \Lc_\varepsilon \Jc_\varepsilon t} \Rc_\varepsilon(\rho_0),\] 
where the second equality follows from the fact that the projector $\Pc_\varepsilon$ commutes with $\Lc_{\varepsilon}$. Defining $\hat{\rho}_0 \equiv \Rc_\varepsilon(\rho_0) \in\check{\Cs}_0$ and \(\hat{\Lc}_\varepsilon \equiv \Rc_\varepsilon \Lc_\varepsilon \Jc_\varepsilon, \) one obtains the reduced model 
$\dot{\hat{\rho}}(t) = \hat{\Lc}_\varepsilon(\hat{\rho}(t)),$ where we have introduced the symbol $\,\hat{\cdot}\,$ to distinguish the generator obtained through AE from the one obtained via the algebraic MR procedure described in Sec.\,\ref{sec:algebraic_reduction}.

\subsubsection{Solving the adiabatic elimination equations.}
The block-spectral decomposition that delivers the reduced model can be efficiently computed through a series expansion that generalizes standard Hamiltonian spectral perturbation techniques to block-diagonalization and Lindblad super-operators. The above framework provides a particularly concise formulation and clarifies the role of coordinate choice within the remaining slow-operator block.

Specifically, with the present notation, an iterative procedure as described in Appendix \ref{app:adiabatic_elimination} can compute the reduction and injection maps $\Rc_\varepsilon$ and $\Jc_\varepsilon$ at different orders and, with those, obtain the reduced generator $\hat{\Lc}_\varepsilon$. 
The resulting set of equations, i.e. Eq.\,\eqref{eq:adel:ordersk}, is under-determined, reflecting the free, $\varepsilon$-dependent isomorphism between $\Ds_\varepsilon$ and $\check{\Cs}_{0}$, thus the gauge degree of freedom \cite{azouit2017towards,riva2024cdc}. The latter involves an interplay between the map $\Jc_\varepsilon$ and the dynamical generator $\hat{\Lc_\varepsilon}$, on which our reformulation of the problem provides added insight.

Specifically, let us focus on the leading orders $\varepsilon^0$ and $\varepsilon^1$. Consistent with $\Cs_0$ being ``simple to analyze'', we assume we have determined factorization maps $\Rc_{(0)}$ and $\Jc_{(0)}$ for the projector $\Pc_{(0)}$. To finalize the order $\varepsilon^0$, we can take the invariance condition and left-multiply by $\Rc_{(0)}$ both sides, obtaining $\hat{\Lc}_{(0)} = \Rc_{(0)}\Lc_{(0)}\Jc_{(0)}$. From Theorem \ref{thm:Lindblad reduction}, we know that $\hat{\Lc}_{(0)}$ is Lindblad. 
Similarly, by taking the invariance condition from Eq.\,\eqref{eq:adel:ordersk} at order $\varepsilon^1$, left-multiplying by $\Rc_{(0)}$ and recalling that $\Rc_{(0)}\Lc_{(0)} = \hat{\Lc}_{(0)}\Rc_{(0)}$, we obtain:
\begin{equation}
    \hat{\Lc}_{(1)} =  \Rc_{(0)}\Lc_{(1)}\Jc_{(0)} + \big[ \hat{\Lc}_{(0)}, \Rc_{(0)}\Jc_{(1)} \big]\, .
    \label{eq:ad_el_order_one}
\end{equation}
While the above depends manifestly on both unknowns $\hat{\Lc}_{(1)}$ and $\Jc_{(1)}$, there are at least two situations in the AE literature where $\Jc_{(1)}$ drops out from Eq.\eqref{eq:ad_el_order_one}:
\begin{itemize}
\item $\Cs_0=\ker\Lc_{(0)}$, as in \cite{azouit2017towards}: this means there are no rotating points (only stationary operators) in the center manifold for $\varepsilon=0$, hence $\hat{\Lc}_{(0)} = 0$.
\vspace*{-1mm}
\item $\Rc_{(0)}\Jc_{(1)}=0$, as in \cite{riva2024cdc}: this corresponds to a particular gauge choice, shown to be valid in their context.
\end{itemize} 
In the above cases, Eq.\,\eqref{eq:ad_el_order_one} directly expresses $\hat{\Lc}_{(1)}$ from known elements, which can then be plugged into the un-projected $\varepsilon^1$-equation in \eqref{eq:adel:ordersk} to determine $\Jc_{(1)}$ (in the second case, together with the associated constraint $\Rc_{(0)}\Jc_{(1)}=0$). Notably, the following result proves that this two-step procedure, solving \eqref{eq:ad_el_order_one} then \eqref{eq:adel:ordersk}, is in fact {\em well-posed} also in the general case, at least for orders $\varepsilon^0$ and $\varepsilon^1$:

\begin{theorem}[Freedom at first-order AE]
\label{thm:ad_el_order_1}
With reference to the iterative setting explained above, we have \(\hat{\Lc}_{(1)} =  \Rc_{(0)}\Lc_{(1)}\Jc_{(0)} + \big[ \hat{\Lc}_{(0)}, \Rc_{(0)}\Jc_{(1)} \big]\), where $\Rc_{(0)}\Jc_{(1)}$ is a free variable.
\end{theorem}

\noindent 
A proof is given in Appendix \ref{app:adiabatic_extra}. This implies that the reduced generator can be computed {\em independently} of the part of $\Jc_{(1)}$ that determines $\Pc_{(1)}$, hence one can first compute $\hat{\Lc}_{(1)}$ based on an arbitrary choice of $\Rc_{(0)}\Jc_{(1)}$, and then substitute it into Eq.\,\eqref{eq:adel:ordersk} to determine $\Jc_{(1)}$. 

\begin{remark}[Limitations of AE] Unlike the projection performed in Sec.\,\ref{sec:algebraic_reduction}, AE attains a MR whose accuracy can in principle be pushed, by solving Eqs.\,\eqref{eq:adel:ordersk} to high-enough order, arbitrarily close to the true behavior of the perturbed system initialized in $\Ds_\varepsilon$. As a tradeoff, AE does {\em not} guarantee that $\hat{\Lc_\varepsilon}$ is Lindbladian, however. Indeed, unlike $\Pc_0$, the projector $\Pc_\varepsilon$ is not  {\em a priori} guaranteed to be CPTP. Hence, proving whether (and when) one obtains reduced dynamics in Lindblad form has been the subject of significant effort in the AE literature, with mixed results \cite{PhysRevA.109.062206}. In \cite{AzouitThesis,zanardi2016dissipative}, a Lindbladian $\Lc_{\varepsilon=0}$ with static center manifold $\Cs_0$, perturbed by a Hamiltonian, is shown to {\em always} yield $\hat{\Lc_\varepsilon}$ Lindblad at order $\varepsilon^2$, modulo a proper gauge choice. The authors of \cite{forni2019cdc,forni2018adiabatic} have established further special conditions in this sense, at the next order and for a non-stationary center manifold $\Cs_0$. Conversely, situations have been identified where conventional gauge choices do not lead to Lindblad dynamics, at higher order \cite{EssigAllOrder} or for non-stationary $\Cs_0$ \cite{riva2024cdc}. This feature is best highlighted in  \cite{PhysRevA.109.062206}, where MR to a qubit-like $\check{\Cs}_0$ is considered, and parameter regimes are identified for which, provably, {\em no} gauge choice can make $\hat{\Lc_\varepsilon}$ Lindblad at order $\varepsilon^4$.
\end{remark}

\subsubsection{Guaranteeing CPTP dynamics at first order.}
Our approach allows one to directly establish the existence of Lindblad generators that solve the AE equations to the first order. In fact, the truncated generator at order $\varepsilon^1$ reads 
\begin{equation*}
    \Tilde{\Lc}_\varepsilon = \Rc_{(0)}(\Lc_{(0)} + \varepsilon\Lc_{(1)})\Jc_{(0)} + \varepsilon [\hat{\Lc}_{(0)}, \Rc_{(0)}\Jc_{(1)}], 
\end{equation*}
along with the truncated reduction and injection maps $\Tilde{\Rc}_\varepsilon \equiv \Rc_{(0)}+\varepsilon\Rc_{(1)}$ and $\Tilde{\Jc}_\varepsilon \equiv \Jc_{(0)}+\varepsilon\Jc_{(1)}$, and the reduced approximate initial ``state'' $\Tilde{\rho}_\varepsilon(0) = \Tilde{\Rc}_\varepsilon(\rho_0)$. Note that, if the original generator is $\Lc_\varepsilon = \Lc_{(0)}+\varepsilon\Lc_{(1)},$ and we assume that $\Rc_{(0)}\Jc_{(1)}$ obeys $[\hat{\Lc}_{(0)},\Rc_{(0)}\Jc_{(1)}]=0$, then the generator obtained from the algebraic MR of Sec.\,\ref{sec:algebraic_reduction} {\em coincides} with the reduced generator obtained from AE truncated at order $\varepsilon^1$, that is, $\check{\Lc}_\varepsilon = \Tilde{\Lc}_\varepsilon$. This allows us to obtain the following sufficient condition:

\begin{corollary}[First-order AE Lindblad dynamics]
\label{cor:ad_el_cptp}
Let $\Lc_\varepsilon=\Lc_{(0)}+\varepsilon\Lc_{(1)}$ be a Lindblad generator for all $\varepsilon\in\Sc$, and let $\Pc_{(0)}, \Rc_{(0)}$ and $\Jc_{(0)}$ be as in Theorem \ref{thm:asymptotic_reduction}. Then, if $[\hat{\Lc}_{(0)},\Rc_{(0)}\Jc_{(1)}]=0$, the generator $\Tilde{\Lc}_\varepsilon$ obtained from AE truncated to order $\varepsilon^1$ is also a Lindblad generator for all $\varepsilon\in\Sc$. 
\end{corollary}

\noindent
The proof of this result follows trivially from Corollary \ref{cor:projected_perturbed_dynamics}. Note that Corollary \ref{cor:ad_el_cptp} extends previous results on the Lindblad property of AE at order $\varepsilon^1$. In particular, in Ref.\,\cite{azouit2017towards} the authors prove that, under the assumption that $\hat{\Lc}_{(0)}=0$, $\Tilde{\Lc}_{\varepsilon}$ is Lindblad regardless of the choice of $\Rc_{(0)}\Jc_{(1)}$. Here, we extend this result to scenarios where $\hat{\Lc}_{(0)}\neq0$, characterizing a set of choices for $\Rc_{(0)}\Jc_{(1)}$ (which includes $\Rc_{(0)}\Jc_{(1)} = 0$), such that $\Tilde{\Lc}_{\varepsilon}$ remains in Lindblad form.  On the other hand, as stated, Corollary \ref{cor:ad_el_cptp} only provides a {\em sufficient} condition: it is possible that, for certain $\Rc_{(0)}\Jc_{(1)}$ with $[\hat{\Lc}_{(0)},\Rc_{(0)}\Jc_{(1)}]\neq0$, $\Tilde{\Lc}_\varepsilon$ remains nonetheless in Lindblad form.

\begin{remark}
While having $\Tilde{\Lc}_\varepsilon$ in Lindblad form ensures that the corresponding reduced model can be interpreted as a standalone quantum system, its relation to the original model may require more: we may want the reduced ``state'' $\Tilde{\rho}_\varepsilon=\Tilde{\Rc}_\varepsilon(\rho)$ to be a valid density operator for any initialization, and any possible final state in the reduced model to be mapped back to a density operator via $\Tilde{\Jc}_\varepsilon(\Tilde{\rho}_\varepsilon)$. This property has not been proved in general in the AE context. In contrast, the algebraic MR approach guarantees these properties, since the corresponding reduction and injection maps $\Rc_{(0)}$ and $\Jc_{(0)}$ can be {chosen} to be CPTP.
\end{remark}

An important question naturally arises at this point, if we consider the {\em error} induced by the approximation, i.e., $\min_{\Rc_{(0)}\Jc_{(1)}} \big|\big| e^{\Lc_\varepsilon t}\Pc_\varepsilon - \Tilde{\Jc}_\varepsilon e^{\Tilde{\Lc}_\varepsilon t}  \Tilde{\Rc}_{\varepsilon} \big|\big| $: Given that $\Rc_{(0)}\Jc_{(1)}$ can be chosen arbitrarily, is there a choice that {\em minimizes} this error? While obtaining a complete answer is beyond our present scope, some intuition may be gained by numerically exploring how different choices of $\Rc_{(0)}\Jc_{(1)}$ affect the error incurred by the reduced dynamics. We address this issue by example in the next subsection, for the same model we considered previously, which also helps us to showcase other aspects of the connections between algebraic and adiabatic approaches and the results we established.

\subsection{Illustrative applications}
\smallskip

\emph{Example 1: Model reduction in a perturbed dissipative XXZ spin chain.---} Let the dissipative XXZ spin-chain model introduced in Sec.\,\ref{sec:xxz_time_crystal} define our unperturbed model, $\Lc_{(0)}$. Consider a purely Hamiltonian perturbation of the form $\Lc_{(1)}(\cdot) = -i[H_1,\cdot]$ whereby, physically, the parameters acquire site-dependent disorder: 
\begin{align*}
    H_1 &= \sum_{j=1}^N \bigg [\frac{\omega_j}{2} \sigma_z^{(j)} + A_{x,j} \sigma_x^{(j)}\sigma_x^{(j+1)} + A_{y,j} \sigma_y^{(j)}\sigma_y^{(j+1)} + A_{z,j} \sigma_z^{(j)}\sigma_z^{(j+1)} \bigg ],
\end{align*}
with values of $A_{x,j},A_{y,j},A_{z,j}$ and $\omega_j$ sampled from a uniform distribution over  $[0,2)$.

Thanks to the simple form of the perturbation, the reduced perturbed generator can be readily computed, namely, $\check{\Lc}(\cdot) = \Rc \Lc \Jc(\cdot) = -i[\check{H}_\varepsilon,\cdot],$ with the new reduced Hamiltonian given by $\check{H}_\varepsilon = \Jc^\dag( H_0 + \varepsilon H_1) = \check{H}_{\frac{N}{2},\varepsilon} \oplus 0_{(N-1)\times (N-1)},$ where \(\check{H}_{\frac{N}{2}} = \frac{1}{2}(\omega_0+\varepsilon\bar{\omega})\sigma_z\)
and $\bar{\omega} \equiv \sum_{j=1}^N \omega_j.$ 

\begin{figure*}[t]
    \centering
    \centering
    \begin{subfigure}[c]{0.49\textwidth}
        \includegraphics[width=\linewidth]{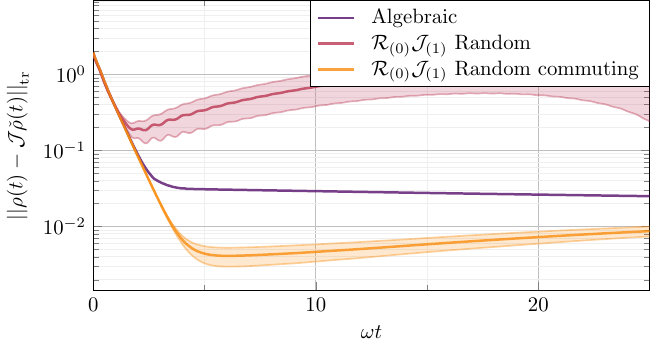}
    \end{subfigure}
    \begin{subfigure}[c]{0.49\textwidth}
        \includegraphics[width=\linewidth]{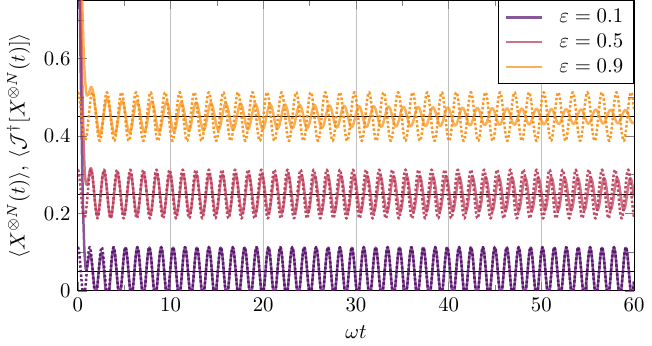}
    \end{subfigure}
    \caption{{\small
    Left: Trace-norm error between the exact and the reduced evolution obtained through algebraic MR and AE techniques for $\Rc_{(0)}\Jc_{(1)}$ randomly chosen and randomly chosen within the commutant $\{\hat{\Lc}_{(0)}\}'$, respectively, with $\varepsilon=0.1$. In the former case, $\Rc_{(0)}\Jc_{(1)}$ was constructed as a matrix with elements sampled from a uniform distribution $\Uc([0,1])$, while in the latter case, it was constructed by picking random uniform values for the diagonal in the same basis as $\hat{\Lc}_{(0)}$. The shaded areas represent the region included in one standard deviation from the average (darker line) of 30 runs of the simulation.  
    Right: synchronization decay as a function of $\varepsilon$. The continuous line represents the evolution of the original model, i.e. $\expect{X^{\otimes N}(t)}$ while the dotted line represents the expectation value predicted by reduced model, i.e. $\expect{\Jc^\dag[X^{\otimes N}(t)]} = \tr\{\Jc^\dag[X^{\otimes N}] \check{\rho}(t)\}$. The curves have been shifted vertically by $\varepsilon/2$ for stylistic reasons. In both figures, $\rho_0=\ketbra{+}{+}^{\otimes N}$, all the other parameters are as in Fig.\,\ref{fig:red_xxz_crystal}. }}
    \label{fig:gauge_choice}
\end{figure*}

\begin{figure*}[t]
    \centering
    \begin{subfigure}[c]{0.49\textwidth}
        \includegraphics[width=\linewidth]{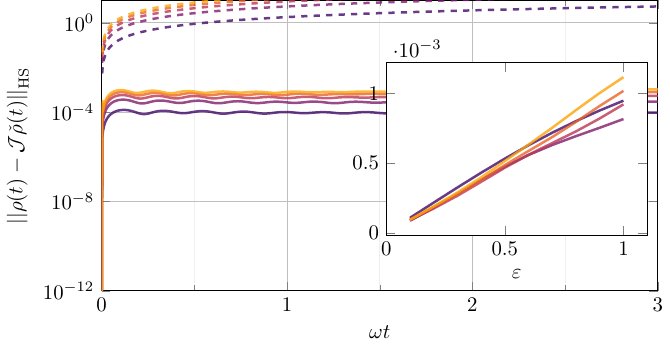}
    \end{subfigure}
    \begin{subfigure}[c]{0.49\textwidth}
        \includegraphics[width=\linewidth]{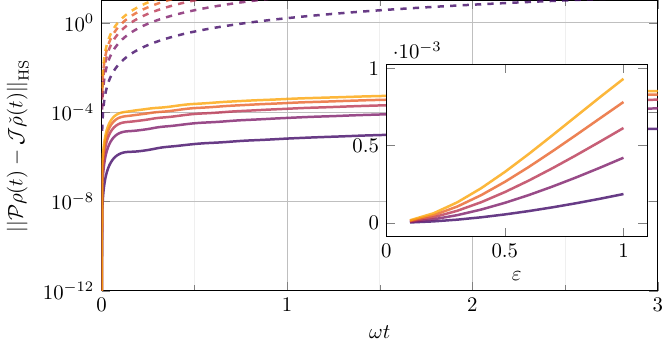}
    \end{subfigure}
    \caption{\small
    Hilbert-Schmidt distance between the reduced model $\Jc\check{\rho}(t)$ obtained through algebraic MR and the exact evolution $\rho(t)$ (left) and the projected exact evolution $\Pc\rho(t)$ (right). The colors represent different values of $\varepsilon$ from $0.1$ to $0.9$ (from dark to light), with increments of $0.2$. For each color, the continuous line represents the Hilbert-Schmidt distance computed via simulation, whereas the dashed line represents the error bound found in Eq.\,\eqref{eqn:bound_1} (left) and \eqref{eqn:bound_2} (right). The insets show the dependence of the same distance with respect to $\varepsilon$ for different values of $\omega t$ from $0.3$ to $2.7$ (from dark to light), with increments of $0.6$. Here, $\rho_0$ is picked at random from $\Cs_0$ and all the other parameters are identical to those in Fig.\,\ref{fig:gauge_choice}. The super-operator norm $\norm{\Lc_{(1)}}_\text{sop}$ used to compute the bounds is the one induced by the Hilbert-Schmidt norm $\norm{\cdot}_{HS}$, i.e. the largest eigenvalue of a matrix representation of $\Lc_{(1)}$  and, since $\Pc=\Pc^\dag$, $C=1$.}
    \label{fig:perturbation_error}
\end{figure*}

In Fig.\,\ref{fig:gauge_choice}(left), we show the error incurred by the reduced dynamics, quantified in terms of the trace-norm distance $||\rho(t) - \Tilde{\Jc}_\varepsilon \Tilde{\rho}_{\varepsilon,t} ||_{\tr}$, for both the algebraic MR and the AE reduced model under different gauge choices. This choice is seen to influence significantly the error behavior. In particular, picking $\Rc_{(0)}\Jc_{(1)}$ at random tends to perform very poorly, while picking $\Rc_{(0)}\Jc_{(1)}$ at random subject to the constraint that $[\hat{\Lc}_{(0)},\Rc_{(0)}\Jc_{(1)}]=0$ yields a much better performance. We further notice that, at longer times, the algebraic MR performs worse than the best cases of AE. This is expected, since the projector $\Pc_{(0)}$ is fixed and does not track the evolution of $\Ds_\varepsilon$. In fact, this difference may be precisely thought of as quantifying the cost of imposing the CPTP requirement in the reduced algebraic model. 

Note that the generator used in the simulation of the orange and purple curve in Fig.\,\ref{fig:gauge_choice} is the same, since for the orange curve $\Rc_{(0)}\Jc_{(1)}$ was chosen to be such that $[\Lc_{(0)},\Rc_{(0)}\Jc_{(1)}] = 0.$ The only difference is thus given by the reduction and injection map used, i.e., $\Rc_{(0)},\Jc_{(0)}$ for the purple curve and $\Tilde{\Rc}_{\varepsilon},\Tilde{\Jc}_{\varepsilon}$ for the orange curves. This further highlights the importance of the projection map used to compute the reduced model. In Fig.\,\ref{fig:gauge_choice}(right), we can see that when the model is perturbed, the synchronization that was previously observed is lost in the long-time regime. Interestingly however, the synchronization persists in the short- to intermediate-time regime (at first order), depending on the strength of the perturbation $\varepsilon$. This is due to the fact that the evolution leaves the unperturbed center manifold, and is further confirmed by Fig.\,\ref{fig:perturbation_error}, where the error trajectory and the respective bounds given in Eqs.\,\eqref{eqn:bound_1}, \eqref{eqn:bound_2} for the reduced algebraic model are shown. In particular, in the insets one observes that the dependence of the distance on the perturbation strength $\varepsilon$ is {\em linear} for $\norm{\rho(t)-\Jc\check{\rho}(t)}_{HS}$ and {\em quadratic} for $\norm{\Pc\rho(t)-\Jc\check{\rho}(t)}_{HS}$.

\medskip

\noindent
\emph{Example 2: Model reduction in a perturbed spin chain with dephasing.---} As a second example, inspired by the analysis in \cite[Chapter 3]{zhang2024driven}, we consider a purely dissipative spin-$1/2$ chain, again under periodic boundary conditions, for which a dominant dephasing dynamics is perturbed by both unital and non-unital transverse dissipation. That is, the Lindblad generator is given by \(\Lc_\varepsilon = \Lc_{(0)} + \varepsilon\Lc_{(1)} , \) with $\Lc_{(0)} \equiv -i[H_0,\cdot] + \sum_{j=1}^N \Dc_{\sigma_z^{(j)}}$, $H_0 \equiv A_z \sum_{j=1}^N \sigma_z^{(j)}\sigma_z^{(j+1)}$, $A_z \in {\mathbb R}$, and $\Lc_{(1)} \equiv \sum_k \Dc_{L_k}$, where the noise operators $L_k$ are
\begin{align*}
    L_{x,j} = \mu_{x,j} \sigma_x^{(j)}, \quad 
    L_{+,j} = \mu_{+,j} \sigma_+^{(j)}, \quad
    L_{-,j} = \mu_{-,j} \sigma_-^{(j)}, \quad \mu_{x,j},\mu_{+,j},\mu_{-,j} \geq 0, \:\forall j.
\end{align*}
In this case, the unperturbed center manifold corresponds to the diagonal algebra $\Cs_0 = \ker\Lc \equiv \Span\{\ketbra{s}{s}\}$, where $s$ denotes all the bitstrings of length $N$. Let then $\Pc$ be the CPTP projector onto the diagonal ${\bf p}=[{\bf p}_1,\ldots {\bf p}_n]^\top=\diag (\rho)$, i.e., $\Pc(\rho) = \sum_s \ketbra{s}{s} \rho \ketbra{s}{s},$ and $\Rc,\Jc$ its two factors
\begin{align*}
    \Rc : &\:\Cb^{n\times n} \mapsto \Cb^n, 
    \quad \rho \mapsto \Rc(\rho) = \sum_s \ketbra{s}{s}\rho\ket{s} ,  \\
    \Jc : &\:\Cb^{n} \mapsto \Cb^{n\times n}, 
   \quad  {\bf p} \mapsto \Jc({\bf p}) = \sum_s \ketbra{s}{s}{\bf p}_s,
\end{align*}
where $n\equiv 2^N$. Notably, by projecting $\Lc_{(1)}$ onto $\Cs_0$ we obtain, similarly to \cite{Benoist_2021}, a {\em classical} model of the form \( \dot{\bf p}(t) = L {\bf p}(t)\), where ${\bf p}\equiv\Rc({\rho})$ is a probability distribution (i.e. $[\bf p]_j\geq0$ and ${\bf 1}^T{\bf p}=1$) and $L$ is a Metzler matrix (i.e., $[L]_{j,k}\geq0$, $\forall j\neq k$), with ${\bf 1}^T L = 0$. One can further note that, in this case, $\Lc_{(1)}$ leaves the steady-state manifold $\Cs_0$ invariant.

For concreteness, let us first study the action of noise operators acting only on the first qubit, i.e., $L_{x,1},L_{+,1}, L_{-,1}.$ If we denote by $0*$ any bitstring of length $N$ that starts with $0$ and, correspondingly, by $1*$ the identical one with the first bit flipped, we have that  
\begin{align*}
    \sum_{u=x,+,-}\Dc_{L_{u,1}}(\ketbra{0*}{0*}) &= \alpha_1(\ketbra{1*}{1*} - \ketbra{0*}{0*}), \quad \alpha_1 \equiv \mu_{x,1}^2 + \mu_{-,1}^2 , \\
    \sum_{u=x,+,-}\Dc_{L_{u,1}}(\ketbra{1*}{1*}) &= \beta_1(\ketbra{0*}{0*} - \ketbra{1*}{1*}) , \quad \beta_1 \equiv \mu_{x,1}^2 + \mu_{+,1}^2. 
\end{align*}
The action of $L_{x,1},L_{+,1}, L_{-,1}$ onto $\Cs_0$ can thus be described by
$$\begin{bmatrix}
    -\alpha_1&\beta_1\\
    \alpha_1&-\beta_1
\end{bmatrix}\otimes \one_{2^{N-1}}.$$
One can verify that such a matrix is indeed Metzler, as $\alpha_1,\beta_1\geq0$. Combining the effect of the dissipators on all spins, we obtain the reduced approximate generator in the Metzler-matrix form
\begin{equation}
    L = \varepsilon \bigg(\sum_{j=1}^N \one_{2^{j-1}} \otimes \begin{bmatrix}
        -\alpha_j&\beta_j\\\alpha_j&-\beta_j
    \end{bmatrix}\otimes \one_{2^{N-j}}\bigg).
    \label{eq:Metzler_matrix}
\end{equation}
Since in this case $\hat{\Lc}_0=0$, the AE generator truncated at order $\varepsilon^1$ coincides with the generator described by the above equation.
Thus, the only difference between the two dynamics stems from the different reduction and injection maps.

\begin{figure*}[t]
    \centering
    \begin{subfigure}[c]{0.49\textwidth}
        \includegraphics[width=\linewidth]{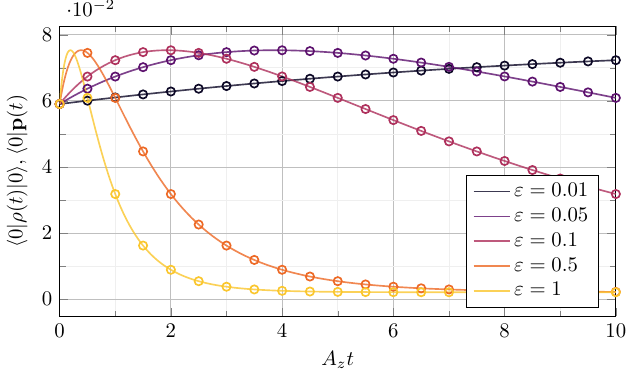}
    \end{subfigure}
    \begin{subfigure}[c]{0.49\textwidth}
        \includegraphics[width=\linewidth]{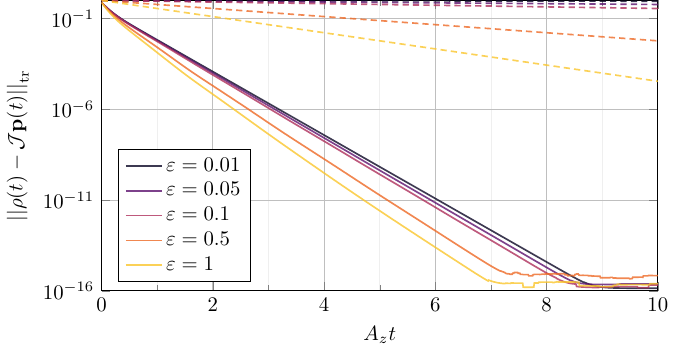}
    \end{subfigure}
    \vspace*{-2mm}
    \caption{\small{Simulation results for the perturbed purely dissipative spin chain, with $N=5$, random initial condition $\rho_0$, and a random set of parameters $\mu_{u,j},$ $u=x,+,-$. Left: Probability of measuring $\ketbra{0\dots0}{0\dots0}$ versus time for different $\varepsilon$. The continuous line represents the full model, $\bra{0}\rho(t)\ket{0}$, whereas the empty dots represent the reduced model, $\bra{0}{\bf p}(t)$. Right: Convergence of the trace-norm distance between $\rho(t)$ and $\Jc {\bf p}(t)$ for different $\varepsilon$. The dashed lines represent exponential convergence, according to $\Gamma e^{- \Delta_{\Lc_\varepsilon}t}$, for different $\varepsilon$. }}
    \label{fig:dissipative}
\end{figure*}

Numerical simulation results are shown in Fig.\,\ref{fig:dissipative}(left). One can observe that, because of the special structure of the perturbation $\Lc_{(1)}$, there is {\em no} error in the probability of measuring $\ketbra{0\dots0}{0\dots0}$. In fact, there is no error also in measuring all the others states $\ketbra{s}{s}$ (we plotted only the first probability solely for ease of presentation). On the other hand, in Fig.\,\ref{fig:dissipative}(right), one can notice a fast convergence of the error norm, which is due to the exponential decay of the coherences.

\begin{remark}[Structure of the reduced model]
Note that, if we were to represent the reduced state ${\bf p}_t$ as a density operator $\rho(t)$, it would seem as if there is no apparent MR, as one still needs a density operator of size $2^N \times 2^N$. On the other hand, the MR becomes obvious when representing the reduced state as a probability vector, since this makes it apparent that the coherences of the state are not relevant for the dynamics. The same ``hidden'' reduction happens for more general block-diagonal matrices, since the MR we achieve does not come only from the presence of repeated blocks ($\otimes\one_{G,k}$ terms). In particular, when the center manifold $\Cs$ has the block structure $\Cs=\bigoplus_k \Bf(\Hc_{F,k})$, the only parts of the state that are relevant for the dynamics are the blocks along the diagonal, and the off-diagonal blocks can be discarded. Representing the state as a single density operator is then only a matter of representation convenience and smaller representations could be proposed. 
\end{remark}

\section{Conclusions and outlook}
\smallskip

In this paper, we have shown how the operator-algebraic framework developed for exact MR of Markovian quantum dynamics may be brought to bear in developing {\em approximate} MR procedures to characterize the long-time dynamics on a center manifold. Depending on whether the Lindblad dynamics defining the full model has a center manifold that is directly suitable for the intended reduction or can be seen as an analytic perturbation of a Lindblad dynamics that does, we have shown how to achieve MR in a way that is asymptotically exact or makes natural contact with AE techniques. Specifically, by realizing how the ``gauge degree of freedom'' appearing in AE equations can be reformulated as a factorization choice for the projector in the algebraic method, we have elucidated how the two MR methods can be made to match at first order, and thus share their respective advantages and limitations. Application-relevant scenarios in which the proposed methods can be employed include systems weakly coupled to structured environments, systems that admit approximately-preserved quantum information encodings, and Lindblad generators in which exact or approximate symmetries exist (for example, with a Hamiltonian and a dissipative generator that commute, or do so approximately).

By design, the algebraic method returns a reduced generator that is in Lindblad form -- hence retains the CPTP properties of a valid, {\em quantum} reduced model; but, in the perturbative setting we consider, it maintains a fixed order of error with respect to the exact evolution. In contrast, AE can in principle be carried to higher orders of approximation and improved error performance, but with the risk of unavoidably losing the CPTP property. In this context, our results provide further insight on the issues associated to the CPTP requirement in AE-reduced models \cite{PhysRevA.109.062206}. Our analysis also highlights a tradeoff between accuracy and structure of the model, something that is already well known in the transition from general non-Markovian master equations to the Lindblad form \cite{breuer2002theory}. In particular, we highlight how using the available degrees of freedom that ensure a CPTP reduction need not result in the smallest possible reconstruction error. In general, for numerical prediction of a target quantum state or observable, a non-CPTP approximation may be preferable; for physical interpretation, composability, and implementation on quantum simulators, a CPTP approximation is needed or otherwise advantageous. Further exploration of this tradeoff is left for future work. 

We have illustrated the applicability of our methods in analyzing the non-stationary long-time dynamics arising in many-body spin-chain models of interest for non-equilibrium physics and quantum simulation. We expect they will prove useful in other applications of interest, notably, in the context of quantum reservoir engineering, where AE is routinely employed \cite{Robin1}. Likewise, our techniques should also apply beyond the finite-dimensional setting we focused on, subject to appropriate truncation and mathematical modifications \cite{Robin2}. This may be especially relevant to bosonic models that exhibit single or multiple limit cycles \cite{li2024exact, ben2021quantum, Bruder2025}, with emphasis on quantum synchronization studies \cite{vdP1,vdP2}, as well as to bosonic quantum error correction \cite{Terhal_2020}.

A natural question for future work would be to extend the algebraic treatment of perturbed models to higher orders of approximation. This may also enable making contact with the idea of a ``dissipation-projected hierarchy'' put forward in \cite{zanardi2016dissipative} for the case where the unperturbed center manifold $\Cs_0$ includes only steady states. In our context, some reduced model in a space isomorphic to $\Cs_0$ does exist, and it should ``mirror'' the one obtained by higher-order AE, in an appropriate sense. However, it is currently unclear how to compute such model with an approach that would preserve the structural CPTP constraint, as the algebraic MR approach guarantees. 

\section{Acknowledgments}
\smallskip

The authors are indebted to Michiel Burgelman and Yikang Zhang for many useful discussions relevant to various aspects of this work.  F.T. acknowledges funding from the European Union - NextGenerationEU, within the National Center for HPC, Big Data and Quantum Computing (Project No.\,CN00000013, CN 1, Spoke 10). F.T and T.G. are partially supported by the Italian Ministry of University and Research under the PRIN project ``Extracting essential information and dynamics from complex networks", grant No.\,2022MBC2EZ. A.S. is partially supported by Plan France 2030 through the project No.\,ANR-22-PETQ-0006. Work at Dartmouth was supported in part by the US Army Research Office under grant No.\,W911NF2210004.

\appendix

\section{Mathematical complements}
\subsection{Model reduction via operator algebras}

\smallskip
\label{app:aymptotic_model_reduction}

In this subsection we recall relevant results from the literature about CPTP projectors and their images \cite{lindblad1999general,wolf2012quantum,ticozzi2017alternating,bialonczyk2018application,blume2010information,wolf2010inverseeigenvalueproblemquantum}, as well as model-reduced dynamics on distorted algebras \cite{grigoletto2024exactmodelreductioncontinuoustime}. Our first two results, contained in Proposition \ref{prop:structure_of_fixed_points} and Theorem \ref{thm:Lindblad reduction}, holds in general, without any assumption on the relationship between the projector $\Pc$ and the generator $\Lc$. A relation between $\Pc$ and $\Lc$ emerges once the properties the reduced dynamics should retain are specified. In Sec.\,\ref{sub: RCM}, these general results are combined with the projection onto the center manifold to construct a reduced Lindblad generator.

\begin{proposition}[Fixed-point structure]
\label{prop:structure_of_fixed_points}
     Let $\Pc$ be a CPTP projector and let $\As$ be its image, i.e., $ \As \equiv {\rm Im}\Pc \subseteq\Bf(\Hc)$. 
    Let $\Hc_R$ be the Hilbert space where $\As$ has no support, in such a way that $\Hc = \supp(\As)\oplus\Hc_R$.
    Then there exist a decomposition of the Hilbert space, 
    \begin{equation}
        \Hc  = 
        \bigoplus_k (\Hc_{F,k}\otimes\Hc_{G,k}) \oplus \Hc_R,
    \end{equation}
    a unitary operator $U\in\Bf(\Hc)$, and a set of density operators $\tau_k\in\Df(\Hc_{G,k})$ such that 
    \begin{equation}
        \As = U\bigg[ \bigoplus_k \big(\Bf(\Hc_{F,k}) \otimes \tau_k\big) \oplus \zero_R \bigg] U^\dag.
        \label{eq:wedderburn_decomposition}
    \end{equation}
    Furthermore, $\As$ is a $*$-algebra closed with respect to the product $X\cdot_\sigma Y \equiv X \sigma Y$, where \[\sigma = U\bigg[ \bigoplus_k(\one_{F,k}\otimes \tau_k^{-1})\oplus \tau_R \bigg]U^\dag,\]  
    with $\tau_R$ being an arbitrary state on $\Hc_R$.
    Let then $W_k$ be a linear operator from $\supp(\As)$ onto $\Hc_{F,k}\otimes\Hc_{G,k}$, such that $W_k W_k^\dag=\one_{F,k}\otimes \one_{G,k}$. Then, $\Pc$ restricted to $\supp(\As)$ takes the form
    \begin{align}
         \Pc|_\As(X) 
            &= U\bigg\{ \bigoplus_{k} \Big[\tr_{\Hc_{G,k}}\left(W_k X W_k^\dag \right)\otimes \tau_k \Big]  \bigg\} U^\dag
        \label{eqn:state_ext_blocks}
    \end{align}
    and is orthogonal with respect to the inner product $\inner{X}{Y}_{\sigma,\lambda} = \tr[X^\dag \sigma^\lambda Y \sigma^{1-\lambda}]$ for $\lambda\in[0,1]$.
\end{proposition}

\noindent 
A proof of this result can be found in the above-mentioned references. In the literature, a $*$-algebra closed with respect to a modified product is also referred to as a \textit{distorted algebra} \cite{blume2010information}, while the block structure provided in Eq.\,\eqref{eq:wedderburn_decomposition} is known as the \textit{Wedderburn decomposition}.

The above representation theorem provides us with a structure for both the algebra and the CPTP projection on the support of the fixed states. From this structure, one can observe that a distorted algebra $\As$ is isomorphic to a {\em smaller} representation, which we denote $\check{\As}$, where the weighted copies are removed, i.e., $\check{\As} \equiv \bigoplus_k \Bf(\Hc_{F,k})\subseteq\Bf(\check{\Hc})$, with $\check{\Hc}=\bigoplus_k \Hc_{F_k}$. This allows us to find a factorization in terms of two maps \cite{grigoletto2023modelreductionquantumsystems}, 
$\Rc:\Bf(\Hc)\to\Bf(\check{\Hc})$, ${\rm Im}\Rc=\check{\As}$ and $\Jc:\check{\As}\to\Bf(\Hc)$, ${\rm Im}\Jc={\As}$, such that 
$$\Rc\Jc = \Ic_{\check{\As}}, \qquad \Pc = \Jc\Rc.$$ 
Furthermore, one can observe \cite{grigoletto2023modelreductionquantumsystems} that these two factors can be chosen to be CPTP, i.e., 
\begin{equation}
\label{eqn:reductionb}
\Rc(X)=\bigoplus_k \tr_{\Hc_{G,k}}(W_k X W_k^\dag)=\bigoplus_k X_{F,k}=\check X , \quad \forall X\in\Bf(\Hc), 
\end{equation}
and for all $\check{X}\in\check{\As}$, $\check{X}=\bigoplus_k X_{F,k}$,
\begin{equation}
\label{eqn:injectionb}
\Jc(\check X)=U\bigg(\bigoplus_k X_{F,k} \otimes\tau_{k}\bigg)U^\dag.
\end{equation}
We name $\Rc$ the \textit{reduction super-operator} and $\Jc$ the \textit{injection super-operator}. These reduction and injection super-operators naturally allow us to compute the reduced dynamics $\check{\Lc} \equiv \Rc\Lc\Jc$ restricted to the algebra $\check{\As}$.  To provide the initial condition for the reduced model in $\As,$ we consider the application of ${\Pc}$ to the initial state. From now on, we will assume that we first reduce the dynamics to $\supp(\As),$ and then consider the compressed description we have just introduced. Importantly, we can prove that the reduced generator $\check{\Lc}$ is a valid Lindblad generator by leveraging the structure of the reduction and injection maps. We here report \cite[Theorem 4]{grigoletto2024exactmodelreductioncontinuoustime} for completeness:

\begin{theorem}[Reduced Lindbladian on algebras]
\label{thm:Lindblad reduction} Let $\Pc$ be a CPTP projection, with $\mathscr{A}$ its image $\As = {\rm Im }\Pc \subseteq \mathcal{B(H)}$, and let $\Rc$ and $\Jc$ be the CPTP factorization of $\Pc = \Jc\Rc$, as defined in Eqs.\,\eqref{eqn:reductionb} and \eqref{eqn:injectionb}. Then, for any Lindblad generator $\Lc$, its restriction to $\check{\As}$, $\check{\Lc}\equiv\Rc\Lc\Jc, $ is also a Lindblad generator, that is, $\check{\Lc}:\check\As\to\check\As$ and $\{e^{\check{\Lc}t}\}_{t\geq 0}$ is a quantum dynamical semigroup. 
\end{theorem}

This theorem ensures that, {\em regardless} of the image of the projector $\Pc$, as long as $\Pc$ is CPTP, the reduced model, $\dot{\check{\rho}} (t)=\check{\Lc}(\check{\rho}(t))$, is in Lindblad form. In order to compute the reduced Hamiltonian and noise operators that enter $\check{\Lc}$, given a representation $(H, \{ L_k\})$ of $\Lc$, one can use the procedure described in \cite[Appendix]{grigoletto2025quantummodelreductioncontinuoustime}.

\begin{remark}[Beyond exact MR]
With the above algebraic tools in hand, the key issue becomes how to choose a sub-algebra $\As$, such that the reduced model is actually a good approximation of the original one. In \cite{grigoletto2024exactmodelreductioncontinuoustime,grigoletto2023modelreductionquantumsystems}, a framework for {\em exact} MR is developed, in which knowledge of the initial states or observables of interest is leveraged to construct minimal sub-algebras that ensure perfect reconstruction of the original dynamics.
It is important to stress that Theorem \ref{thm:Lindblad reduction} only provides a \emph{sufficient} condition for quantum MR. This means that, while we can guarantee that $\check{\Lc}$ is Lindblad whenever the projector $\Pc$ is CPTP, it is possible to find reduced generators that are Lindblad even if this is not the case. Despite their appeal, exact reductions may be too stringent a requirement for obtaining useful models in settings of practical relevance, however.  This motivates the approaches we develop in Sections 3 and 4 in the main text.

\end{remark}

\subsection{The adiabatic elimination equations}
\label{app:adiabatic_elimination}

\smallskip

AE \cite{azouit2017towards} aims to compute the dynamics $\hat{\Lc}_\varepsilon$ on $\Ds_\varepsilon$, as well as the perturbed injection map $\Jc_\varepsilon$, by using a Taylor expansion in $\varepsilon$. $\hat{\Lc}_\varepsilon$ and $\Jc_\varepsilon$ can be computed directly, without the need to obtain the reduction map $\Rc_\varepsilon$. This is done by leveraging the so-called  \textit{invariance condition}: since $\Ds_\varepsilon$ is invariant under $\Lc_\varepsilon$, we have that \(\Lc_\varepsilon \Pc_\varepsilon = \Pc_\varepsilon \Lc_\varepsilon \Pc_\varepsilon\). By substituting $\Pc_\varepsilon$ with its factorization $\Pc_\varepsilon = \Jc_\varepsilon\Rc_\varepsilon$ and right-multiplying by $\Jc_\varepsilon$, one obtains the invariance condition as given in \cite{azouit2017towards}:
\begin{equation}
    \Lc_\varepsilon\Jc_\varepsilon = \Jc_\varepsilon \hat{\Lc}_\varepsilon .
    \label{eq:invariance_condition}
\end{equation} 
To solve Eq.\,\eqref{eq:invariance_condition}, both $\Jc_\varepsilon$ and $\hat{\Lc}_\varepsilon$ are expanded into their respective Taylor series, i.e., 
\begin{align*}
    \Jc_\varepsilon &= \sum_{k=0}^{+\infty} \varepsilon^k \Jc_{(k)}, &\hat{\Lc}_\varepsilon&= \sum_{k=0}^{+\infty} \varepsilon^k \hat{\Lc}_{(k)}. 
\end{align*}
Substituting these two expansions into Eq.\,\eqref{eq:invariance_condition} yields 
\begin{equation*}
   \sum_{k=0}^{+\infty} \varepsilon^k \bigg(\sum_{j=0}^k \Lc_{(j)}\Jc_{(k-j)}\bigg) =  \sum_{k=0}^{+\infty} \varepsilon^k \bigg(\sum_{j=0}^k \Jc_{(j)}\hat{\Lc}_{(k-j)}\bigg).
\end{equation*}
Since the equality must hold for all $\varepsilon$, it must hold separately for every order $\varepsilon^k$. This results in a set of equations to solve:
\begin{align}
    \nonumber \varepsilon^0&&\Lc_{(0)}\Jc_{(0)} &= \Jc_{(0)}\hat{\Lc}_{(0)},\\
    \nonumber \varepsilon^1&&\Lc_{(1)}\Jc_{(0)} + \Lc_{(0)}\Jc_{(1)} &= \Jc_{(1)}\hat{\Lc}_{(0)} +  \Jc_{(0)}\hat{\Lc}_{(1)},\\
   \label{eq:adel:ordersk} \vdots&&\vdots \\
   \nonumber  \varepsilon^k&& \sum_{j=0}^k \Lc_{(j)}\Jc_{(k-j)} &= \sum_{j=0}^k \Jc_{(j)}\hat{\Lc}_{(k-j)} \; .
\end{align}
The main idea is to solve this set of equations iteratively, starting at order $\varepsilon^1$ and solving for $\Jc_{(1)}$ and $\hat{\Lc}_{(1)}$, then increasing each in order, until one stops at a given maximum order, say, $K$, to have a truncated approximation of accuracy $\varepsilon^K$ \cite{azouit2017towards}. At order $\varepsilon^k$, assuming that all the equations have been solved up to order $\varepsilon^{k-1}$ means that all the terms $\Jc_{(j)}$ and $\hat{\Lc}_{(j)}$ have been computed for all $j=0,\dots,k-1$, so the resulting equation contains $\hat{\Lc}_{(k)}$ and $\Jc_{(k)}$ as the only unknowns.

\section{Proofs and technical results}
\smallskip

\label{app:proofs}

\begin{proof} \textbf{[Lemma \ref{lem:expo_convergence}]}
Let us decompose $\rho_0$ into $\rho_0 = \sum_j a_jX_j + \sum_j b_jY_j$, where $X_j,Y_j$ form a basis for the generalized eigenspaces of $\Lc$ and $X_j\in\Cs$. Thanks to the spectral properties of the Lindblad generator $\Lc$, we have 
$$ e^{\Lc t}\rho_0 = \sum_{j} a_je^{i\omega_j t } X_j + \sum_j b_je^{\lambda_j t} p_j(t) Y_j,\quad \Re(\lambda_j)<0, $$ 
\noindent 
with $p_j(t)$ some polynomial in $t$. Since $\Pc$ acts trivially on $\Cs$, we have that $e^{\Lc t}\rho_0 - \Pc(e^{\Lc t}\rho_0) = \sum_j b_j e^{\lambda_j t} p_j(t) Y_j$. The norm is then bounded by $\norm{e^{\Lc t}\rho_0 - \Pc(e^{\Lc t}\rho_0)} \leq \sum_{j} |b_j| e^{\lambda_j t} |p_j(t)| \norm{Y_j}.$ Every exponential is bounded by $e^{-\Delta_\Lc t}$, i.e., $e^{\lambda_jt} \leq e^{-\Delta_\Lc t}$ for all $j$. Furthermore we have that, for all $\delta>0$,  $|p_j(t)|e^{\lambda_j t}\leq|p_j(t)|e^{-\Delta_\Lc t}\leq \Gamma_j e^{-(\Delta_\Lc-\delta)t}$ where $\Gamma_j$ can be taken to be $\Gamma_j \geq \max_{t\geq0} |p_j(t)|e^{-\delta t}$, which clearly exists and is non-negative. To conclude, we have \sloppy $\sum_{j} |b_j| e^{\lambda_j t} |p_j(t)| \norm{Y_j} \leq e^{-(\Delta_\Lc-\delta) t}\sum_{j} |b_j|  \Gamma_j \norm{Y_j} \equiv \Gamma e^{-(\Delta_\Lc-\delta) t} $.
\end{proof}

\begin{proof} \textbf{[Proposition \ref{prop:spectral_projection}]}
This proof uses the Dirichlet's Theorem on Diophantine approximation, which we recall for completeness: Given a set of real values $\alpha_1,\dots,\alpha_d\in\Rb$ and a natural number $k\in\Nb$, there exists a set of integers $p_1,\dots,p_d\in\Zb$ and a natural $q\in\Nb$, $1\leq q\leq k^d$ such that $|\alpha_j q - p_j|\leq 1/k$. 
    
Let us denote by $\omega_j\in\Rb$ the imaginary part of the purely imaginary eigenvalues corresponding to the eigen-operators in $X\in\Cs$, i.e., $\Lc(X) = i\omega_j X$ for some $X\in\Cs$, and fix $\alpha_j = \frac{\omega_j}{2\pi}$.  
This, with Dirichlet's Theorem, allows us to construct an increasing sequence $\{q_k\}_{k\in\Nb}\subset\Nb$, and $d$ sets of integers $\{p_{j,k}\}_{k\in\Nb}\subset\Zb$, one for each $j=1,\dots,d$, such that $\lim_{k\to+\infty} |\frac{\omega_j}{2\pi}q_k-p_{j,k}| = 0$. 
Using the sequence $\{q_k\}_{k\in\Nb}$ we then define the projector $\Pc$ as 
    \[\Pc \equiv \lim_{k\to+\infty} e^{q_k \Lc}.\]
Clearly, $\Pc$ is CPTP, since $\Lc$ is a Lindblad generator, and $[\Pc,\Lc]=0$ by construction. It remains to prove that $\Pc$ is a projector and that ${\rm Im}\Pc = \Cs$. 

Let us then consider the eigen-operators of $\Lc$, i.e., $X\in\Bf(\Hc)$ such that $\Lc(X) = (-\mu+i\omega) X$, with $\mu\geq0$ and $\omega\in\Rb$. We can divide the operator space $\Bf(\Hc)$ into the direct sum of the center manifold $\Cs$ and an operator subspace containing all of the other eigen-operators, i.e., $\Fs \equiv \Span\{X\in\Bf(\Hc)| \Lc(X)=(-\mu+i\omega) X, \mu>0\}$ and 
    \(\Bf(\Hc) = \Cs \oplus \Fs.\)
    Then, if $\mu >0$, i.e., $X\in\Fs$ we have that $e^{\Lc q_k}(X) = e^{-\mu q_k}e^{i\omega q_k} X$, hence $\lim_{k\to+\infty} e^{\Lc q_k}(X) = 0$. On the other hand, if $\mu=0$, i.e., $X\in\Cs$, we have $e^{\Lc q_k}(X) = e^{i\omega_j q_k}X = e^{i 2\pi \alpha_j q_k }X = e^{i 2\pi (\alpha_j q_k - p_{j,k}) } e^{i2\pi p_{j,k}}X$; therefore,  
    \[\lim_{k\to+\infty} e^{\Lc q_k}(X) = \lim_{k\to+\infty} 
    e^{i 2\pi 
    (\alpha_j q_k - p_{j,k})}
    e^{i2\pi p_{j,k}}
    X 
    = X, \]
    since $(\alpha_j q_k - p_{j,k}) {\to 0}$ and $p_{j,k}\in\Zb$. 
    This implies that, $\forall X\in\Cs$ $\Pc(X) = X$, while $\forall X\in\Fs$, $\Pc(X)=0$, proving that $\ker\Pc = \Fs$ and ${\rm Im}\Pc = \Cs$ and the fact that $\Pc$ is a projector.  
\end{proof}

\begin{proof} \textbf{[Theorem \ref{thm:asymptotic_reduction}]}
The proof of the first point follows directly from the invariance of the center manifold, i.e., $\Lc\Cs\subseteq\Cs$. In fact, for any $\rho\in\Cs$ and any $t\geq0$, we have that $e^{\Lc t}\rho\in\Cs$; hence, $e^{\Lc t}(\rho) = \Pc e^{\Lc t}(\rho) = \Jc \Rc e^{\Lc t} \Jc\Rc(\rho) = \Jc e^{\check{\Lc} t} \Rc(\rho)$.

The proof of the second point comes from the attractivity of the center manifold. Let us define $\Qc \equiv \Ic - \Pc$ and note that $[\Qc,\Lc]=0$. We can then decompose $\Lc$ into $\Lc = \Pc\Lc\Pc + \Qc\Lc\Qc$ where the mixed terms $\Pc\Lc\Qc=\Qc\Lc\Pc=0$ since $[\Pc,\Lc]=[\Qc,\Lc]=0$. Then, since $[\Qc\Lc\Qc,\Pc\Lc\Pc]=0$ we have that $e^{\Lc t} = e^{\Pc\Lc\Pc t} e^{\Qc\Lc\Qc t}$ and since $\Qc\Lc\Qc$ contains only the eigenvalues of $\Lc$ with strictly negative real part we have $\lim_{t\to+\infty} e^{\Qc\Lc\Qc t} = \Pc$ hence 
$$\lim_{t\to+\infty} \big( e^{\Lc t} - \Pc e^{\Lc t}\big) = 0. $$ 
Since $\Jc e^{\check{\Lc} t} \Rc = \Pc e^{\Lc t}$, this concludes the proof of the second point. 

The proof of the third point follows from the exponential convergence of the center manifold given in Lemma \ref{lem:expo_convergence}.

The proof of the last point follows as a specialization to the continuous-time case of \cite[Theorem 8, point 2]{wolf2010inverseeigenvalueproblemquantum}.
Let us assume $U=\one$ for the sake of simplicity, as the arguments that follow hold for any $U$. Define $ \Ec \equiv \lim_{k\to+\infty} e^{(q_k-1)\Lc}$ with $q_k$ the sequence derived in the proof of Proposition \ref{prop:structure_of_fixed_points}. Then $\Ec$ is CPTP and $\Ec e^\Lc = \Pc$. This implies that $e^\Lc$ maps {\em extreme states} in $\Cs$ into extreme states in $\Cs$, where by {extreme states in $\Cs$ we mean states} that can not be written as a convex combination of other states in $\Cs$. To prove this fact consider an extreme state $\rho\in\Cs$, then $e^\Lc(\rho) = \sum_k \lambda_k \sigma_k$ with a nontrivial convex sum of states $\sigma_k$. Applying $\Ec$ on both sides we obtain $\Ec e^{\Lc}\rho = \Pc(\rho) = \rho = \sum_k \lambda_k \Ec(\sigma_k)$ which is a contradiction. 

Extreme states $\rho\in\Cs$ have the form $\rho = \bigoplus_k \delta_{j,k}\ketbra{\phi}{\phi} \otimes \tau_k  \in\Cs$ for some $j$. Then, for any extreme state $\rho\in\Cs$ we have $e^{\Lc t}(\rho) = \bigoplus_k \delta_{j',k}\ketbra{\phi'}{\phi'} \otimes \tau_k $ for some $j'$. By continuity in time of the evolution it must be $j=j'$. This directly implies that, for any $\rho\in\Cs$, we have 
\[e^{\Lc t}(\rho) = \bigoplus_k \Ec_{t,k} (\rho_k)\otimes\tau_k, \]
where $\Ec_{t,k}$ are a set of CPTP quantum dynamical semigroup that map extreme states into extreme states, i.e., $\Ec_{k,t}$ are unitary conjugations $\Ec_{t,k} = e^{\Lc_{(k)} t}$  with $\Lc_{(k)}(\rho_k) = -i[H_{F,k},\rho_k]$ and $H_{F,k}\in\Bf(\Hc_{F,k})$. The result then follows by summing the effects of each generator in the different blocks.  
\end{proof}

\begin{proof} \textbf{[Theorem \ref{thm:error_bounds}]}
Let us expand both exponentials in their Taylor expansion:
\begin{align*}
  &\norm{\sum_{k=0}^{+\infty} \frac{t^k}{k!} [\Lc_\varepsilon^k - \check{\Lc}_\varepsilon^k]\Pc_{(0)} }_{\rm sop}  
  = \norm{t [\Lc_\varepsilon\Pc_{(0)} - \Pc_{(0)}\Lc_\varepsilon\Pc_{(0)}]  + O(t^2)}_{\rm sop}
   \leq t\norm{\Qc_{(0)}\Lc_\varepsilon\Pc_{(0)}}_{\rm sop}   + O(t^2),
\end{align*}
where $\Qc_{(0)} \equiv \Ic - \Pc_{(0)}$. Defining $\bar{\Lc}_{(0)} \equiv \Lc_\varepsilon - \Lc_{(0)}$, and noting that $\Pc_{(0)}$ commutes with $\Lc_{(0)}$, it follows that 
\begin{align*}
\Qc_{(0)}\Lc_\varepsilon\Pc_{(0)} &= \Qc_{(0)}\Lc_{(0)}\Pc_{(0)} + \Qc_{(0)}\bar{\Lc}_{(0)}\Pc_{(0)}
=\cancel{\Qc_{(0)}\Pc_{(0)}}\Lc_{(0)}\Pc_{(0)} + \Qc_{(0)}\bar{\Lc}_{(0)}\Pc_{(0)}
=\Qc_{(0)}\bar{\Lc}_{(0)}\Pc_{(0)}.
\end{align*}
We can then observe that, since $\Pc_{(0)}$ and $\Qc_{(0)}$ are projectors, we have $\norm{\Pc_{(0)}}_{\rm sop} = \norm{\Qc_{(0)}}_{\rm sop} \equiv C\geq1$ (see Ref.\,\cite{szyld2006many}), and thus $\norm{\Qc_{(0)}\bar{\Lc}_{(0)}\Pc_{(0)}}_{\rm sop} \leq C^2\norm{\bar{\Lc}_{(0)}}_{\rm sop} $.
Furthermore, we have that $\norm{\bar{\Lc}_{(0)}}_{\rm sop} \leq \sum_{k=1}^{+\infty} |\varepsilon|^k \norm{\Lc_{(k)}}_{\rm sop} = |\varepsilon|\norm{\Lc_{(1)}}_{\rm sop} + O(\varepsilon^2)$, implying that 
\[\norm{\left(e^{t\Lc_\varepsilon} - e^{ t\check{\Lc}_\varepsilon}\right)\Pc_{(0)}}_{\rm sop} \leq t C^2 |\varepsilon| \norm{\Lc_{(1)}}_{\rm sop} + O(t^2,\varepsilon^2),\]
whereby the bound in Eq.\,\eqref{eqn:bound_1} immediately follows.

Similarly, we have:
\begin{align*}
  \norm{\sum_{k=0}^{+\infty} \frac{t^k}{k!} \Pc_{(0)}[\Lc_\varepsilon^k - \check{\Lc}_\varepsilon^k]\Pc_{(0)} }_{\rm sop}  \leq t & \norm{[\cancel{\Pc_{(0)}\Lc_\varepsilon\Pc_{(0)}} - \cancel{\Pc_{(0)}\Lc_\varepsilon\Pc_{(0)}}]}_{\rm sop} \\&\qquad + \frac{t^2}{2} \norm{\Pc_{(0)}[\Lc_\varepsilon^2 - \check{\Lc}_\varepsilon^2]\Pc_{(0)}}_{\rm sop} + O(t^3).
\end{align*}
Then, with some algebraic manipulations, one can show that 
\begin{align*}
    \Pc_{(0)}[\Lc_\varepsilon^2 - \check{\Lc}_\varepsilon^2]\Pc_{(0)} & = \varepsilon^2\Pc_{(0)}[\Lc_{(1)}^2 - \Lc_{(1)}\Pc_{(0)}\Lc_{(1)}]\Pc_{(0)} 
    + O(\varepsilon^3),
\end{align*}
and thus  
\begin{align*}
    \norm{\Pc_{(0)}[\Lc_\varepsilon^2 - \check{\Lc}_\varepsilon^2]\Pc_{(0)}}_{\rm sop} \leq C^3\varepsilon^2 \norm{\Lc_{(1)}}_{\rm sop}^2+ O(\varepsilon^3),
\end{align*}
which concludes the proof. Note that in this proof we used the fact that super-operator norms are sub-multiplicative, i.e., $\norm{\Ac\Bc}_{\rm sop}\leq \norm{\Ac}_{\rm sop}\norm{\Bc}_{\rm sop}$ and that projectors have unit norm, i.e., $\norm{\Pc}_{\rm sop}=\norm{\Qc_{(0)}}_{\rm sop}\geq 1$. Since these properties hold for any superoperator norm, the theorem holds for any super-operator norm. Finally, note that in case $\Pc$ is orthogonal with respect to the inner product that defines the desired superoperator norm (e.g., if $\Pc = \Pc^\dag$ and $\norm{\cdot}_{\rm sop}$ is the norm induced by the Hilbert-Schmidt inner product), then $C=1.$
\end{proof}

\label{app:adiabatic_extra}

\begin{proof}\textbf{ [Theorem \ref{thm:ad_el_order_1}]}
Consider the commutative condition 
\(\Lc_\varepsilon \Pc_\varepsilon = \Pc_\varepsilon \Lc_\varepsilon\) and recall that $\Pc_\varepsilon$ is idempotent, i.e., $\Pc_\varepsilon^2 = \Pc_\varepsilon$ and that both $\Lc_\varepsilon$ and $\Pc_\varepsilon$ are analytic. Then, by expanding both $\Lc_\varepsilon$ and $\Pc_\varepsilon$ into their Taylor expansion, and separating different orders of $\varepsilon$, we obtain the couple of equations that must be solved at all orders:
\begin{equation*}
\varepsilon^k:\quad\begin{cases}
        \sum_{j=0}^k \Lc_{(j)}\Pc_{(k-j)} = \sum_{j=0}^k \Pc_{(j)}\Lc_{(k-j)}\\
        \sum_{j=0}^k \Pc_{(j)}\Pc_{(k-j)} = \Pc_{(k)}
    \end{cases}.
\end{equation*}
Observe that, assuming we have solved all the orders up to $\varepsilon^{k-1}$, at order $\varepsilon^k$ we are left to solve two equations in a single unknown,  the projector $\Pc_{(k)}.$ Once $\Pc_{\varepsilon}$ or one of its finite-order approximation has been determined, one can proceed to determine the factorizations $\Jc_\varepsilon$ and $\Rc_\varepsilon$ such that $\Rc_\varepsilon \Jc_\varepsilon = \Ic$ and $\Jc_\varepsilon \Rc_\varepsilon = \Pc_\varepsilon$. Expanding into Taylor expansion, we obtain, at order $\varepsilon^k$ (for $k>0$), the following two constraints: 
\[ \varepsilon^k:\quad
\begin{cases}
    \Pc_{(k)} = \sum_{j=0}^k \Jc_{(j)}\Rc_{(k-j)} \\
    \sum_{j=0}^k \Rc_{(j)}\Jc_{(k-j)}=0
\end{cases}, \]
while, at order $\varepsilon^0$, the second constraint becomes $\Rc_{(0)}\Jc_{(0)} = \Ic$. In this case, if we assume to have solved the problem up to order $\varepsilon^{k-1}$, then at order $\varepsilon^k$ we must  solve a problem with two equations and two unknowns. In this proof we are only going to focus on orders $\varepsilon^0$ and $\varepsilon^1$, and at each order we solve the above two sets of equations one after the other.

$\bullet$ \textit{Order $\varepsilon^0$.}
At order $\varepsilon^0$, we have to solve $\Pc_{(0)}\Lc_{(0)} = \Lc_{(0)}\Pc_{(0)}$ and $\Pc_{(0)}^2 = \Pc_{(0)}$. Then, there exists a basis (which depends on the eigen-decomposition of the unperturbed generator $\Lc_{(0)}$) such that $\Lc_{(0)}$ and $\Pc_{(0)}$ can be written in the following block form 
\begin{equation*}
\Lc_{(0)} \equiv \left[\begin{array}{c|c}
    L_0&0\\\hline
    0&L_S
\end{array}\right], \qquad 
\Pc_{(0)} \equiv \left[\begin{array}{c|c}
    X_A^{0}&X_B^{0}\\\hline
    X_C^{0}&X_D^{0}
\end{array}\right] ,
\end{equation*}
where $L_0$ contains only imaginary eigenvalues, $L_S$ contains only eigenvalues with strictly negative real part, and the off-diagonal blocks are $0$, since both $\Cs_0$ and its complement are invariant. Then, the commutation relation imposes 
\begin{align*}
\begin{cases}
[X_A^{0},L_0]=0,\\
X_B^0L_S-L_0X_B^0 = 0,\\
X_C^0L_0-L_SX_C^0 = 0,\\
[X_D^0,L_S] = 0,
\end{cases}
\end{align*}
and the second and third (Sylvester) equations impose $X_B^0 = X_C^0 = 0$. The condition that $\Pc_{(0)}$ should be a projector imposes that $X_A^0$ and $X_D^0$ are projectors. Furthermore, since $\Pc_{(0)}$ should be a projector onto $\Cs_0$, we have that $X_A^0 = \one$ and $X_D^0 = 0$ thus resulting in 
\begin{equation*}
\Pc_{(0)} = \left[\begin{array}{c|c}
    \one&0\\\hline
    0&0
\end{array}\right].
\end{equation*}
We can then focus on computing the factorization of the projector. Let  \(\Rc_{(0)} = \big[\begin{array}{c|c}
    R_A^0&R_B^0
\end{array}\big]\) and \(\Jc_{(0)} = \big[\begin{array}{c|c}
    (J_A^0)^T&
    (J_B^0)^T
\end{array}\big]^T\). 
With these choices, the constraints on the factorization become  
\begin{align*}
   \begin{cases}
       J_A^0R_A^0 = \one ,\\
       J_A^0R_B^0 = 0,\\
       J_B^0R_A^0 = 0,\\
       J_B^0R_B^0 = 0,\\
       R_B^0J_B^0 = 0 .
   \end{cases}
\end{align*}
From the first constraint, we have $J_A^0$ invertible and $R_A^0 = (J_A^0)^{-1}$. The other constraints impose $R_B^0 = 0 = J_B^0$, thus resulting in \(\Rc_{(0)} = \big[\begin{array}{c|c}
    (J_A^0)^{-1}&0
\end{array}\big]\) and \(\Jc_{(0)} = \big[\begin{array}{c|c}
    J_A^0&0
\end{array}\big]^T.\)
The reduced dynamics at order $\varepsilon^0$ is then given by $\hat{\Lc}_{(0)} = \Rc_{(0)}\Lc_{(0)}\Jc_{(0)} = (J_A^0)^{-1} L_0 J_A^0 $ where $J_A^0$ is simply a change of coordinates in the basis of $\Cs_{(0)}$. 

Note that the CPTP property is not invariant with respect to all changes of coordinates in $\Cs_{(0)}$. The natural basis choice in $\Cs_{(0)}$ is determined by the factorization of the CPTP projection $\Pc_{(0)}$: $\Jc_{(0)} T T^{-1}\Rc_{(0)}=\Pc_{(0)}$ represents a valid factorization for any change of base $T$ but not all of them ensure that $\Jc_{(0)} T$ and $T^{-1}\Rc_{(0)}$ are also CPTP. Since the factorization into $\Jc_{(0)}\Rc_{(0)}=\Pc_{(0)}$ given in Eq.\,\eqref{eqn:reduction} and \eqref{eqn:injection} ensures that $\Jc_{(0)}$ and $\Rc_{(0)}$ are also CPTP, it provides the natural choice for a basis in $\Cs_{(0)}$. 
Furthermore we know that, given $\Lc_0$ that is Lindblad, by choosing $J_A^0 = U^T\otimes U$ for some unitary operator $U\in\check{\Cs}_0$, then $\hat{\Lc}_{(0)}$ is also Lindblad. However, for other changes of bases this need not be the case.
Since $J_A^0$ acts only as a change of basis on the evolution and thus does not change the accuracy of the model, one can take $J_A^0 = \one$.

$\bullet$ \textit{Order $\varepsilon^1$.} In the same basis we had at order $\varepsilon^0$, we can write 
\begin{equation*}
\Lc_{(1)} = \left[\begin{array}{c|c}
    L_A^1&L_B^1\\\hline
    L_C^1&L_D^1
\end{array}\right], \qquad 
\Pc_{(1)} = \left[\begin{array}{c|c}
    X_A^1&X_B^1\\\hline
    X_C^1&X_D^1
\end{array}\right].
\end{equation*}
In this case, the commutation condition $\Lc_{(1)}\Pc_{(0)} + \Lc_{(0)}\Pc_{(1)} = \Pc_{(0)}\Lc_{(1)} + \Pc_{(1)}\Lc_{(0)}  $ imposes:
\begin{align*}
&\begin{cases}
    [L_0,X_A^1] = 0,\\
    L_0X_B^1 - X_B^1L_S = L_B^1,\\
    L_SX_C^1 - X_C^1L_0 = -L_C^1, \\
    [L_S,X_D^1] = 0,
\end{cases}
\end{align*}
while the projector condition $\Pc_{(0)}\Pc_{(1)} + \Pc_{(1)}\Pc_{(0)} = \Pc_{(1)}$ imposes $X_A^1=0$ and $X_{D}^1=0$. The second and third Sylvester equations admit a unique solution, since $L_0$ and $L_S$ share no eigenvalues. Let us denote by $X_B^{1*}$ and $X_C^{1*}$ their solution. This results in the term at order $\varepsilon^1$ of the projector being
\begin{align*}
    \Pc_{(1)} = \left[\begin{array}{c|c}
    0&X_B^{1*}\\\hline
    X_C^{1*}&0
\end{array}\right].
\end{align*}

Let \(\Rc_{(1)} = \big[\begin{array}{c|c}
    R_A^1&R_B^1
\end{array}\big]\) and \(\Jc_{(1)} = \big[\begin{array}{c|c}
    J_A^1&
    J_B^1
\end{array}\big]^T.\)
Then, the constraints are $\Rc_{(0)}\Jc_{(1)} + \Rc_{(1)}\Jc_{(0)} = 0$ and $ \Jc_{(1)}\Rc_{(0)} + \Jc_{(0)}\Rc_{(1)} = \Pc_{(1)}$, which become:
\(R_A^1 = -(J_A^0)^{-1}J_A^1(J_A^0)^{-1}\)
and
\begin{align*}
&\begin{cases}
    R_B^1 = (J_A^0)^{-1}X_B^{1*},\\
    J_B^1 = X_C^{1*}J_A^0.
\end{cases}
\end{align*}
This finally results in 
\begin{align*}
&\Rc_{(1)} = \left[\begin{array}{c|c}
    -(J_A^0)^{-1}J_A^1(J_A^0)^{-1}&(J_A^0)^{-1}X_B^{1*} 
\end{array}\right] ,  
&\Jc_{(1)} = \left[\begin{array}{c}
    J_A^1\\\hline
    X_C^{1*}J_A^0, 
\end{array}\right], 
\end{align*}
where $J_A^1$ is a free variable. The term of order $\varepsilon^1$ in the reduced dynamics is then 
\begin{align*}
    \hat{\Lc}_{(1)} &= \Rc_{(0)}\Lc_{(1)}\Jc_{(0)} + [\hat{\Lc}_{(0)},\Rc_{(0)}\Jc_{(1)}]
    = (J_A^0)^{-1} L_A^1 J_A^0 + [(J_A^0)^{-1} L_0 J_A^0, J_A^1].
\end{align*}
As stated in the theorem, the term $\Rc_{(0)}\Jc_{(1)}$ appearing on the right side of the commutator is indeed a free variable. This concludes the proof.

One may note that for  $J_A^0 = \one$, we obtain
\(\Rc_{(1)} = \big[\begin{array}{c|c}
    -J_A^1&X_B^{1*}
\end{array}\big]\), \(\Jc_{(1)} = \big[\begin{array}{c|c}
    J_A^1&
    X_C^{1*}
\end{array}\big]^T\) and \(\hat{\Lc}_{(1)} = L_A^1 + [L_0, J_A^1].\)
At first sight, the choice of $J_A^1$ now appears to modify the dynamics. However, it is actually equivalent to a basis change in $\Cs_0$, up to terms of order $\varepsilon^2$. The latter would appear through the impact of $J_A^1$ on the next order. Indeed, the spectrum of the dynamics on the exact reduced model (infinite Taylor expansion) should be a physical quantity, independent of gauge choices.
\end{proof}

\section{Additional material for the dissipative XXZ spin-chain example}
\label{app:xxz_extra}
\smallskip 

We start by explicitly showing that the angular momentum operator $J_z$ is a conserved quantity and the generator of a group of strong symmetries.

\begin{lemma}
\label{lam:strong_symmetry_crystal}
Let $\Lc$ be the Lindblad generator with the Hamiltonian and noise operators given in Eqs.\,\eqref{eq:xxz_ham} and \eqref{eq:xxz_jump}. Then 
$\Lc^\dag(J_z) = 0$ and $U_\varphi = e^{-i\varphi J_z}$ is a strong symmetry for the model, i.e., $[U_\varphi,H]=[U_\varphi,L_j]=0$. 
\end{lemma}
\begin{proof}
We first prove that $\Lc^\dag(J_z)=0$. The fact that $[H,J_z]=0$ is well-known in the literature, see e.g.,  \cite{franchini2017introduction}. It remains to prove that $\sum_j \Dc^\dagger_{L_j}(J_z)=0$, where \(\Dc^\dagger_{L}(X) \equiv L^\dagger X L - \frac{1}{2}\{L^\dag L,X\}\). One way to show this is to notice that $[J_z, L_j] = [J_z,L_j^\dagger] =0$, for all $j$, with $L_j=\sigma_+^{(j)} \sigma_-^{(j+1)}$. Let us rewrite 
\begin{align*}
    \sigma_+^{(j)} \sigma_-^{(j+1)} 
    &= \frac{1}{4}\big( \vec{\sigma}^{(j)}\cdot\vec{\sigma}^{(k)} -  \sigma_z^{(j)}\sigma_z^{(j+1)} \big) 
     +\frac{i}{4} \big( \sigma_y^{(j)}\sigma_x^{(j+1)} -  \sigma_x^{(j)}\sigma_y^{(j+1)}  \big),
\end{align*}
with $\vec{\sigma}^{(j)}\cdot\vec{\sigma}^{(k)} \equiv \sigma_x^{(j)}\sigma_x^{(k)} + \sigma_y^{(j)}\sigma_y^{(k)} + \sigma_z^{(j)}\sigma_z^{(k)}$. Since the operators $\vec{\sigma}^{(j)}\cdot\vec{\sigma}^{(j+1)}$ are invariant under collective rotations generated by $J_z$ \cite{KLV}, it is immediate to see that the operator following the factor $\tfrac{1}{4}$ above commutes with $J_z$. Direct calculation using the Pauli commutation rules shows that the commutator between $\sigma_z^{(j)} +\sigma_z^{(j+1)}$ (hence $J_z$) and the operator following the $\tfrac{i}{4}$ factor also vanishes. The fact that $[J_z,(\sigma_+^{(j)} \sigma_-^{(j+1)})^\dagger] =0$ may be established in an entirely similar manner. 

Since the continuous group $U_\varphi$ is generated by $J_z$, and the latter operator commutes with $H$ and all the $L_j$, it also follows that $[U_\varphi,H]= [U_\varphi,L_j]=0$.
\end{proof}

Next, we prove that the states $\overline{\rho}_m$ in Eq.\,\eqref{eq:steady_states} are steady states for the dynamics of interest.

\begin{proposition}
\label{prop:steady_states}
    Let $\Lc$ be the Lindblad generator with the Hamiltonian and noise operators given in Eqs.\,\eqref{eq:xxz_ham} and \eqref{eq:xxz_jump} and let $\bar{\rho}_m$ be the states defined in Eq.\,\eqref{eq:steady_states}. Then $\Lc(\bar{\rho}_m) = 0$ for all $m=-\frac{N}{2},\dots,\frac{N}{2}$.
\end{proposition}
\begin{proof}
Consider the decomposition of the Hilbert space into $\{\Hc_m\}$, $m=-N/2,\ldots,N/2$ according to the $J_z$ symmetry. On each $\Hc_m$, the state $\bar{\rho}_m$ is proportional to the identity. Therefore, it trivially commutes with the Hamiltonian part. 

To analyze the action of the dissipator $\mathcal{D}$, we first note that the diagonal in the canonical basis $\{\ket{0},\ket{1}\}^{\otimes N}$, remains invariant under the action of any $\mathcal{D}_{L_j}$. The dynamics induced by $\mathcal{D}$ on this diagonal amounts to a classical Markov chain. Let us introduce the compact notation \(\kket{s} \equiv \ketbra{s}{s}\) and study the effect of $L_1 = \sigma_+^{(1)} \sigma_-^{(2)}$ on $\kket{00*}, \kket{01*}, \kket{10*}, \kket{11*}$, where $*$ collectively denotes the values of the bits $j+2,\ldots,N$ on which $L_1$ acts trivially. We then have 
    \begin{align} \label{app:eq:as01}
    \Dc_{L_1}(\kket{00*}) &= \Dc_{L_1}(\kket{01*}) = \Dc_{L_1}(\kket{11*}) = 0,\\  \label{app:eq:as02}
        \Dc_{L_1}(\kket{10*}) &= \gamma(\kket{01*} - \kket{10*}).
    \end{align}
From there, and the same holding for all $j$, the uniform distribution over all $\kket{s_m}$ is a steady state. Indeed, plugging in this state, we get:
\begin{eqnarray*}
    \Dc_{L}\bigg(\tfrac{1}{d_m} \sum_{s_m} \kket{s_m}\bigg) & = & \tfrac{1}{d_m} \sum_{s_m} \Dc_{L}(\kket{s_m}) = \tfrac{1}{d_m} \sum_j \sum_{s_m} \Dc_{L_j}(\kket{s_m}) \; .
\end{eqnarray*}
In this last expression, each term is either 0 rightaway from \eqref{app:eq:as01}, or contributes a term as in \eqref{app:eq:as02}. On the right-hand side, the term $-\tfrac{\gamma}{d_m} \kket{s_m}$ will appear as many times as $s_m$ has a subsequence of the type $10$; the term $+\tfrac{\gamma}{d_m} \kket{s_m}$ will appear as many times as $s_m$ has a subsequence of the type $01$. When considering a bitstring with indices $N+1=1$, these two numbers of occurrences are equal. Hence, we conclude that $\Dc_{L}\left(\tfrac{1}{d_m} \sum_{s_m} \kket{s_m}\right)=0$.
\end{proof}

Let $\As \equiv \Span\{\ketbra{s}{s}\}$, where $s$ are bit strings of length $N$, denote the commutative algebra of the diagonal elements in the standard basis. As a consequence of the above proposition, we also obtain that, as long as $A_{xy}=0$, the diagonal is an invariant subspace, as formalized in the following:

\begin{corollary}
Let $\Lc$ be the Lindblad generator with the Hamiltonian and noise operators given in Eqs.\,\eqref{eq:xxz_ham} and \eqref{eq:xxz_jump}. For $A_{xy}=0$, the commutative algebra $\As$ is $\Lc$-invariant.    
\end{corollary}

Next, we prove that the coherences $\ketbra{0_L}{1_L}$ and $\ketbra{1_L}{0_L}$ belong to the center manifold.

\begin{proposition}
    Let $\Lc$ be the Lindblad generator with Hamiltonian and noise operators given in Eqs.\,\eqref{eq:xxz_ham} and \eqref{eq:xxz_jump} and 
    let $\ket{0_L} \equiv \ket{00\dots0}$, $\ket{1_L} \equiv \ket{11\dots1},$ respectively. 
    Then, for any value of $A_{xy}$, 
    \begin{align*}
        \Lc(\ketbra{0_L}{1_L}) &= -i(E_0-E_1) \ketbra{0_L}{1_L},\\
        \Lc(\ketbra{1_L}{0_L}) &= i(E_0-E_1) \ketbra{1_L}{0_L},
    \end{align*}
    where $E_0,E_1$ are such that $H\ket{0_L} =E_0\ket{0_L}$ and $H\ket{1_L} = E_1\ket{1_L}$. In particular, $E_0 = (A_z+\frac{\omega}{2})N$ and $E_1 = (A_z-\frac{\omega}{2})N$.
\end{proposition}

\begin{proof}
    Let us start by noticing that, for all $j$, $L_j\ket{0_L} = L_j\ket{1_L} = L_j^\dag\ket{0_L} = L_j^\dag\ket{1_L}=  0$, hence 
    \[\Dc_{L_j}(\ketbra{0_L}{1_L}) = \Dc_{L_j}(\ketbra{1_L}{0_L})=0.\]
    Then, we note that $(\sigma_x^{(j)}\sigma_x^{(j+1)}+\sigma_y^{(j)}\sigma_y^{(j+1)})\, \ket{0_L} = (\sigma_x^{(j)}\sigma_x^{(j+1)}+\sigma_y^{(j)}\sigma_y^{(j+1)})\, \ket{1_L} = 0$ for all $j$. From there the conclusion of the proof follows trivially. 
 \end{proof}

With $\As$ defined as above, let
$\Pc$ be the CPTP projector onto the diagonal ${\bf p}=[{\bf p}_1,\ldots {\bf p}_n]^\top=\diag (\rho)$, i.e., $\Pc(\rho) = \sum_s \ketbra{s}{s} \rho \ketbra{s}{s},$ and $\Rc,\Jc$ its two factors $\Rc(\rho) = \sum_s \ketbra{s}{s}\rho \ket{s}={\bf p}$ and $\Jc({\bf p}) = \sum_s \ketbra{s}{s}{\bf p}_s$. We next prove that, for $A_{xy}=0$, the center manifold in this diagonal algebra is spanned by the states defined in Eq.\,\eqref{eq:steady_states}.
 
\begin{proposition}
Let $\Lc$ be the Lindblad generator with the Hamiltonian and noise operators given in Eqs.\,\eqref{eq:xxz_ham} and \eqref{eq:xxz_jump}.  Then $\dot{\bf p}(t) = L {\bf p}(t),$  where $F \equiv \Rc \Lc \Jc \in\Rb^{n\times n}$ 
is a Metzler matrix (i.e., $[F]_{j,k} \geq 0$, $\forall j\neq k$), with ${\bf 1}_n^T F = 0$ and $n=\dim(\Hc)$. 
Furthermore, any initial state ${\bf p}(0)=\Rc(\rho(0))$ converges under $L$ towards a steady state $\overline{\bf p}\in\Rb^n$, $F\overline{\bf p}=0$, which is partitioned in $N+1$ subvectors associated to the invariant subspaces $\Hc_m$, with each subvector proportional to the uniform distribution ${\bf 1}_{d_m}/d_m$.
\end{proposition}

\begin{proof}
The fact that the reduced model on $\As$ coincides with the continuous-time Markov chain $\dot{\bf p}(t) = F {\bf p}(t)$ is a direct consequence of Theorem \ref{thm:Lindblad reduction} and the fact that, on a commutative algebra, a Lindblad generator becomes a Markov-chain generator, which is represented by Metzler matrix such that ${\bf 1}_n^T F = 0.$ Thanks to the the strong symmetry given in Lemma \ref{lam:strong_symmetry_crystal}, $F$ is the generator of $N+1$ disconnected sub-networks of respective sizes $d_m$, each with a generator $F_m$. The fact that each sub-network is weight-balanced, i.e., $F_m {\bf 1}_m= 0$ for each $m$, has been derived in the proof of Proposition \ref{prop:steady_states}.  Finally, when each sub-network is weight-balanced and connected, by using \cite[Theorem 7.2 and Theorem 7.4]{FB-LNS}, it follows that the state on each sub-network converges exponentially towards the uniform distribution ${\bf 1}_{d_m}/d_m$. 

The fact that each sub-network $m$ is connected can be easily obtained from the following argument. Consider a starting string with $m$ bits on 1 followed by $N-m$ bits on 0. Consider any other end string with $m$ bits on 1. It is easy to work out a path from the starting string to the end string by consecutive swaps from $10$ to $01$ on bits $j,j+1$, as are applied by our Markov chain: first move the rightmost 1 of the starting string towards its position in the end string; then move the second from right; and so on. This concludes the proof.
\end{proof}

\begin{corollary}
Let $\Lc$ be the Lindblad generator with the Hamiltonian and noise operators given in Eqs.\,\eqref{eq:xxz_ham} and \eqref{eq:xxz_jump}, with $A_{xy}=0$ and $\As = \Span\{\ketbra{s}{s}\}$ the commutative algebra of the diagonal elements in the standard basis. Then any state in $\As$ converges to a mixture of the steady states $\bar{\rho}_m$ defined in Eq.\,\eqref{eq:steady_states}.
\end{corollary}

Next, we prove that, as long as $A_{xy}=0$, all the coherences other than $\ketbra{0_L}{1_L}$ and $\ketbra{1_L}{0_L}$ converge to zero, in the canonical basis, implying that the space $\As\oplus\Span\{\ketbra{0_L}{1_L},\ketbra{1_L}{0_L}\}$ is invariant and attractive.

\begin{proposition}
    Let $\Lc$ be the Lindblad generator with the Hamiltonian and noise operators given in Eqs.\,\eqref{eq:xxz_ham} and \eqref{eq:xxz_jump} and let $A_{xy}=0$.
    Then any coherence $\bra{s}\rho(t)\ket{s'}$ with bitstrings $s\neq s'$ apart from $\bra{0_L}\rho(t)\ket{1_L}$ and $\bra{1_L}\rho(t)\ket{0_L}$ converges to zero. 
\end{proposition}

\begin{proof}
    To analyze the action of the dissipator $\Dc$, let us start by fixing $j$ for the sake of convenience. For instance, for $j=1$ we can study the effect of $L_1 = \sigma_+^{(1)} \sigma_-^{(2)}$ on $\ketbra{s}{s'}$, where $s,s'$ are two different bit strings. We start by noting that $L_1^\dag L_1 = \gamma \ketbra{10}{10}\otimes\one_{2^{N-2}}$ and that $L_1 \ket{00*} = L_1\ket{01*} = L_1\ket{11*}=0$ and $L_1\ket{10*} = \sqrt{\gamma}\ket{01*}$ where $*$ collectively denote the values of the bits $j+2,\ldots,N$ on which $L_1$ acts trivially. Then we obtain that $\Dc_{L_1}(\ketbra{s}{s'})=0$ whenever the first two bits of $s$ or of $s'$ (or of both) differ from $10$. In the case where the first two bits of either $s$ or $s'$ (but not both) are equal to $10$, we obtain $\Dc_{L_1}(\ketbra{s}{s'})= - \frac{\gamma}{2} \ketbra{s}{s'}$ and in the case where the first two bits of both $s$ and $s'$ are $10$ we obtain $\Dc_{L_1}(\ketbra{s}{s'}) = \gamma\ketbra{01*}{01*'} - \gamma\ketbra{s}{s'}$ (where $*$ and $*'$ denote the part of the string $s$, resp. $s'$, with the first two qubits removed). 

    \sloppy

    Combining the effects of all the noise operators, we then obtain the effect of the dissipative term $\Dc \equiv \sum_{j=1}^N \Dc_{L_j}$ as 
    \(\Dc(\ketbra{s}{s'}) = \gamma[F(s,s')- \frac{1}{2}g(s,s')\ketbra{s}{s'}]\),
    where 
    \[F(s,s') \equiv \sum_{j=1}^{N} \delta_{s_{j},1}\delta_{s_{j+1},0}\delta_{s_{j}',1}\delta_{s_{j+1}',0} \ketbra{v(s,j)}{v(s',j)},\] 
    with $v(s,j)$ a function that returns the string $s$ in which the characters $j,j+1$ have been changed to $01$; and $g(s,s')\equiv \sum_{j=1}^{N}\delta_{s_{j},1}\delta_{s_{j+1},0}+\delta_{s_{j}',1}\delta_{s_{j+1}',0}\in\Rb$. 

Regarding the Hamiltonian term, for $A_{xy}=0$ the action onto $\ketbra{s}{s'}$ yields $-i[H,\ketbra{s}{s'}] = i\omega_{s,s'}\ketbra{s}{s'}$, where $\omega_{s,s'}\in\Rb$ depends on $s$ and $s'$. Combining this with the effect of the dissipative term $\Dc$, we obtain, for any off-diagonal term $\rho_{s,s'} \equiv \bra{s}\rho \ket{s'}$ with $s\neq s'$:
    \begin{align*}
        \frac{d}{dt}\rho_{s,s'}(t) = &\left(i\omega_{s,s'} - \frac{\gamma}{2} g(s,s')\right)\rho_{s,s'}(t) + \gamma\sum_{w,w'}\bra{s}F(w,w')\ket{s'}\rho_{w,w'}(t),
    \end{align*}
    where $w,w'$ are bit strings.
    This equation implies, for the absolute value of $\rho_{s,s'}$:
    \begin{align*}
        \frac{d}{dt}|\rho_{s,s'}(t)|\leq &- \frac{\gamma}{2}g(s,s')|\rho_{s,s'}(t)| + \gamma\sum_{w,w'}\bra{s}F(w,w')\ket{s'}|\rho_{w,w'}(t)| ,
    \end{align*}
where $\bra{s}F(w,w')\ket{s'}\geq0$. Now, we have that $\bra{s}F(w,w')\ket{s'} \neq 0$ if and only if $s=v(w,j)$ and $s'=v(w',j)$ and $w_jw_{j+1}=10$ and $w_j'w_{j+1}'=10$ for some $j$ or, equivalently, $v^{-1}(s,j)=w$ and $v^{-1}(s',j) = w'$ and $s_{j}s_{j+1}=01$ and $s_{j}s_{j+1}'=01$, where $v^{-1}(s,j)$ is the function that returns the string $s$ with the bits in positions $j,j+1$ changed to $10$. We can thus write $\sum_{w,w'}\bra{s}F(w,w')\ket{s'}|\rho_{w,w'}(t)| = \sum_{j=1}^N f(s,s',j)|\rho_{v^{-1}(s,j),v^{-1}(s',j)}(t)|$ with $f(s,s',j) \equiv \delta_{s_{j},0}\delta_{s_{j+1},1}\delta_{s_{j}',0}\delta_{s_{j+1}',1}$, resulting in 
    \begin{align}
        \frac{d}{dt}|\rho_{s,s'}(t)|\leq &- \frac{\gamma}{2}g(s,s')|\rho_{s,s'}(t)| + \gamma\sum_{j=1}^N f(s,s',j)|\rho_{v^{-1}(s,j),v^{-1}(s',j)}(t)|.
        \label{eq:coherence_dynamics}
    \end{align}
    
From these observations, we can conclude the proof with a Lyapunov function and the LaSalle invariance principle. Consider $V \equiv \sum_{s,s'} |\rho_{s,s'}(t)|$, where the sum runs over all $s\neq s'$ and $s,s'\neq 0_L,1_L$. According to Eq.~\eqref{eq:coherence_dynamics}, the contribution of $|\rho_{w,w'}(t)|$ to $\tfrac{d}{dt}V$ is as follows: (i) For each pair of bits $j,j+1$, there is a positive term, proportional to $\gamma$, if and only if $f(v(w,j),v(w',j),j)=1$, which means $w_jw_{j+1}=10$ and $w_j'w_{j+1}'=10$. But in this case, there is an equal negative term via $-\tfrac{\gamma}{2} g(w,w')$, where $\delta_{w_{j},1}\delta_{w_{j+1},0}+\delta_{w_{j}',1}\delta_{w_{j+1}',0}$ contributes +2 to the value of $g(w,w')$. In other cases, there is no positive term. Hence, $\tfrac{d}{dt}V \leq 0$. (ii) For each pair of bits $j,j+1$, there is a further negative contribution, proportional to $-\tfrac{\gamma}{2}$, whenever $w_jw_{j+1}=10$ or $w_j w_{j+1}'=10$ (but not both: the latter terms have already been used in point (i)). Hence, the Lyapunov function $V$ will strictly decrease, unless $|\rho_{w,w'}(t)|=0$ for all pairs $w,w'$ on which this case happens for some $j$. These are all pairs on which at least one sequence $10$ happens at a different place in $w$ than in $w'$. In other words, the set where $\tfrac{d}{dt}V = 0$ can have $|\rho_{w,w'}(t)|\neq 0$ only for pairs of bitstrings $w,w'$ in which all sequences $10$ happen at the same places; e.g., $w={10}0{10}11$ and $w'={10}1{10}01$. Note that, since we are in a fixed subspace $m$, the numbers of 1's in both bitstrings must also match.

According to the LaSalle's invariance principle, the state will converge to the largest invariant set where $\tfrac{d}{dt}V = 0$. Consider a state where $\tfrac{d}{dt}V = 0$ with $|\rho_{w,w'}(t)|\neq 0$ and, since $w\neq w'$, we can write $w=*100...01*$ and $w'=*100...00*$, where the dots denote a substring of all 0's, of equal length in both $w,w'$, and the $*$ denote other symbols, all equal in both $w,w'$. If $|\rho_{w^{(1)},{w'}^{(1)}}(t)|=0$ for the bitstrings $w^{(1)}=*010...01*$ and ${w'}^{(1)}=*010...00*$ obtained by flipping the first two bits in the identified substring of $w$ and $w'$, then this will not stay so under Eq.\,\eqref{eq:coherence_dynamics}. In search of the invariant set, we can thus assume that $|\rho_{w^{(1)},{w'}^{(1)}}(t)| \neq 0$, writing $w^{(1)}=*100...01*$ and ${w'}^{(1)}=*100...00*$, with now the central dots denoting all 0's on one less bit, and the left $*$ substring including one more bit, compared to $w$ and $w'$. By iterating this argument, we get to a state with $|\rho_{w^{(k)},{w'}^{(k)}}(t)| \neq 0$ where $w^{(k)}=*101*$ and ${w'}^{(k)}=*100*$. The action of $\mathcal{D}_{L_j}$ on the two leftmost identified bits, as expressed through Eq.\,\eqref{eq:coherence_dynamics}, will contribute a positive term to $|\rho_{w^{(k+1)},{w'}^{(k+1)}}(t)|$, where $w^{(k+1)}=*011*$ and ${w'}^{(k+1)}=*010*$. Now, ${w'}^{(k+1)}$ features a sequence $10$ at a place where $w^{(k+1)}$ features $00$. This means that, to be in the set where $\tfrac{d}{dt}V = 0$, we must have $|\rho_{w^{(k+1)},{w'}^{(k+1)}}(t)|=0$. The whole argument thus shows that, in the invariant set where $\tfrac{d}{dt}V = 0$, we cannot have $|\rho_{w,w'}(t)|\neq 0$, because this would ultimately imply $|\rho_{w^{(k+1)},{w'}^{(k+1)}}(t)| \neq 0$, for which $\tfrac{d}{dt}V \neq 0.$ By contradiction, and since $w\neq w'$ were chosen arbitrarily, we must thus conclude that the largest invariant set where $\tfrac{d}{dt}V = 0$ necessarily requires $|\rho_{w,w'}(t)| = 0$ for all $w \neq w'$ and $w,w' \neq 0_L,1_L$. The LaSalle invariance principle thus ensures that any state will asymptotically converge to this situation.
\end{proof}

It follows from the above results that, when $A_{xy}=0$, the kernel and the center manifold are {\em fully} characterized by Eqs.\,\eqref{eq:steady_states} and \eqref{eq:center_manifold} in the main text.

\begin{corollary} 
\label{cor:steady_states_j0}
Let $\Lc$ be the Lindblad generator with the Hamiltonian and noise operators given in Eqs.\,\eqref{eq:xxz_ham} and \eqref{eq:xxz_jump}. If $A_{xy}=0$, then the only equilibria of $\Lc$ are the states $\bar{\rho}_m$ in Eq.\,\eqref{eq:steady_states}, and the center manifold $\Cs$ coincides with Eq.\,\eqref{eq:center_manifold}. 
\end{corollary}

For the kernel of $\Lc$, we can conclude the proof for {\em arbitrary} values of $A_{xy}$ by invoking a standard genericity argument from  \cite[Lemma 5]{ticozzi2013steadystateentanglementengineeredquasilocal}, which we report here for completeness:
\begin{lemma}
\label{lem:generic}
    Let $A_\varphi\in\Rb^{n\times n}$ be analytic in $\varphi\in\Rb$ and let $r=\max_{\varphi\in\Rb}\rank(A_\varphi)$. Then the set $\Omega =\{\varphi\in\Rb| \rank(A_\varphi)<r\}$ is such that $\mu(\Omega)=0$, where $\mu$ is the Lebesgue measure in $\Rb$. 
\end{lemma}

\begin{theorem}
\label{thm:uniqueness_steady_states}
    Let $\Lc$ be the Lindblad generator with the Hamiltonian and noise operators given in Eqs.\,\eqref{eq:xxz_ham} and \eqref{eq:xxz_jump}. 
    For almost all choices of $A_{xy}$, we have that $\ker\Lc$ coincides with the states $\bar{\rho}_m$ defined in Eq.\,\eqref{eq:steady_states}.
\end{theorem}

\begin{proof}
 From Proposition \ref{prop:steady_states} we have that $\dim\ker(\Lc)\geq N+1$, thus $\rank(\Lc)\leq n^2-N-1$. Then from Corollary \ref{cor:steady_states_j0}, for $A_{xy}=0$ we indeed have $\dim\ker(\Lc)=N+1$, thus $\rank(\Lc) = n^2-N-1$. By applying Lemma \ref{lem:generic}, the set of $A_{xy}$ such that $\dim\ker(\Lc)> N+1$ is a measure zero set.
\end{proof}

For the center manifold $\Cs$, we are not aware of an exact equivalent to Lemma \ref{lem:generic}. A perturbative argument could be invoked to show that the conclusions of Corollary \ref{cor:steady_states_j0} remain true at least for $A_{xy}$ sufficiently small. As remarked in the text, numerical evidence suggests they may hold in general.

\smallskip

\bibliographystyle{plain}
\bibliography{ref}

@book{antoulas,
author = {Antoulas, A. C.},
title = {Approximation of Large-Scale Dynamical Systems},
publisher = {Advances in Design and Control, vol. 6;},
year = {2005},
doi = {10.1137/1.9780898718713},
address = {Society for Industrial and Applied Mathematics, Philadelphia},
}

@article{Havel,
	author = {Havel, T. F.},
	doi = {10.1063/1.1518555},
	journal = {J. Math. Phys.},
	pages = {534},
	title = {Robust procedures for converting among {L}indblad, {K}raus and matrix representations of quantum dynamical semigroups},
	volume = {44},
	year = {2003},
}

@article{Dobrynin_2025,
doi = {10.1088/2058-9565/ae0363},
year = {2025},
volume = {10},
pages = {045041},
author = {Dobrynin, D. and Cardarelli, L. and M\"{u}ller, M. and Bermudez, A.},
title = {Compressed-sensing {L}indbladian quantum tomography with trapped ions},
journal = {Quantum Sci. Tech.},
}

@article{Fazio,
	title={{Many-body open quantum systems}},
	author={R. Fazio and J. Keeling and L. Mazza and M. Schir\`{o}},
	journal={SciPost Phys. Lect. Notes},
	volume={99},
	pages={1}, 
	year={2025},
	publisher={SciPost},
	doi={10.21468/SciPostPhysLectNotes.99},
}

@article{Jamir,
	title={{Is {L}indblad for me?}},
	author={M. Stefanini and A. A. Ziolkowska and D. Budker and U. Poschinger and F. Schmidt-Kaler and A. Browaeys and A. 
	Imamoglu and D. Chang and J. Marino},
	journal={SciPost Phys. Lect. Notes},
	volume={129},  
    pages={1},
	year={2026},
}

@book{Weiss,
author = {Weiss, U.},
title = {Quantum Dissipative Systems},
publisher = {World Scientific},
year = {2012},
doi = {10.1142/8334},
address = {},
edition   = {4th},
}

@article{Murch,
	author = {Harrington, P. M. and Mueller, E. J. and Murch, K. W.},
	doi = {10.1038/s42254-022-00494-8},
	journal = {Nat. Rev. Phys.},
	pages = {660},
	title = {Engineered dissipation for quantum information science},
	volume = {4},
	year = {2022},
}

@article{Sieberer,
  title = {Universality in driven open quantum matter},
  author = {Sieberer, L. M. and Buchhold, M. and Marino, J. and Diehl, S.},
  journal = {Rev. Mod. Phys.},
  volume = {97},
  pages = {025004},
  year = {2025},
  doi = {10.1103/RevModPhys.97.025004},
}

@article{ClerkHTR,
  title = {Hidden Time-Reversal Symmetry, Quantum Detailed Balance and Exact Solutions of Driven-Dissipative Quantum Systems},
  author = {Roberts, D. and Lingenfelter, A. and Clerk, A. A.},
  journal = {PRX Quantum},
  volume = {2},
  pages = {020336},
  numpages = {33},
  year = {2021},
  doi = {10.1103/PRXQuantum.2.020336},
}

@article{tokieda2026exact,
      title={Exact diagonalization of a non-quadratic bosonic {L}iouvillian with two-body loss}, 
      author={M. Tokieda},
      year={2026},
      journal={arXiv:2603.27480},
      archivePrefix={arXiv},
      primaryClass={quant-ph},
      url={https://arxiv.org/abs/2603.27480}, 
}

@book{Cohen,
title ={Atom-Photon Interactions: Basic Processes and Applications}, 
author = {C. Cohen-Tannoudji and J. Dupont-Roc and G. Grynberg},
 publisher = {Wiley-VCH},
year = {1998},
doi = {10.1002/9783527617197},
address = {},
}

@article{Altman,
  title = {Quantum Simulators: Architectures and Opportunities},
  author = {Altman, E. and others}, 
  journal = {PRX Quantum},
  volume = {2},
  pages = {017003},
  year = {2021},
  doi = {10.1103/PRXQuantum.2.017003},
}

@article{Cirac,
	author = {Verstraete, F. and Wolf, M. M. and Cirac, J. I.},
	doi = {10.1038/nphys1342},
	journal = {Nat. Phys.},
	pages = {633},
	title = {Quantum computation and quantum-state engineering driven by dissipation},
	volume = {5},
	year = {2009},
}

@phdthesis{AzouitThesis,
  title={Adiabatic elimination for open quantum systems},
  author={Azouit, R.},
  year={2017},
  school={Universit{\'e} Paris sciences et lettres}
}

@article{EssigAllOrder,
  title={Quantum adiabatic elimination at arbitrary order for photon number measurement},
  author={Sarlette, A. and Rouchon, P. and Essig, A. and Ficheux, Q. and Huard, B.},
  journal={{IFAC}-Papers{O}nLine},
  volume={53},
  pages={250},
  year={2020},
  publisher={Elsevier}
}

@ARTICLE{letter2024,
  author={Grigoletto, T. and Ticozzi, F.},
  journal={IEEE Control Sys. Lett.}, 
  title={Exact Model Reduction for Discrete-Time Conditional Quantum Dynamics}, 
  year={2024},
  volume={8},
  pages={550},
  doi={10.1109/LCSYS.2024.3399100}
  }

@article{Schirmer2010,
  title = {Stabilizing open quantum systems by {M}arkovian reservoir engineering},
  author = {Schirmer, S. G. and Wang, X.},
  journal = {Phys. Rev. A},
  volume = {81},
  pages = {062306},
  year = {2010},
  doi = {10.1103/PhysRevA.81.062306},
}

@book{alicki-lendi,
	author = {R. Alicki and K. Lendi},
	publisher = {Springer-Verlag},
	series = {Lecture Notes in Physics},
	title = {Quantum Dynamical Semigroups and Applications},
	volume = {286},
	year = {1987}
}

@misc{wolf2012quantum, 
	author = {M. M. Wolf},
	title = {Quantum {C}hannels and {O}perations - {G}uided {T}our},
	note = {{L}ecture {N}otes available at \emph{https://mediatum.ub.tum.de/download/1701036/ 1701036.pdf}}, 
	year = {2012}
}

@article{Ticozzi2007QuantumMS,
  title={Quantum Markovian Subsystems: Invariance, Attractivity, and Control},
  author={F. Ticozzi and L. Viola},
  journal={IEEE Trans. Autom. Control},
  year={2007},
  volume={53},
  pages={2048},
  url={https://api.semanticscholar.org/CorpusID:14850189}
}

@article{ticozzi2017alternating,
  title={Alternating projections methods for discrete- time stabilization of quantum states},
  author={Ticozzi, F. and Zuccato, L. and Johnson, P. D. and Viola, L.},
  journal={IEEE Trans. Autom. Control},
  volume={63},
  pages={819},
  year={2017},
  publisher={IEEE}
}

@article{johnson2015general,
  title={General fixed points of quasi-local frustration-free quantum semigroups: from invariance to stabilization},
  author={Johnson, P. D. and Ticozzi, F. and Viola, L.},
  journal={Quantum Inf. Comput.},
  year={2016},
  pages={0657},
  volume={16}, 
}

@article{blume2010information,
  title = {Information-preserving structures: {A} general framework for quantum zero-error information},
  author = {Blume-Kohout, R. and Ng, H.-K. and Poulin, D. and Viola, L.},
  journal = {Phys. Rev. A},
  volume = {82},
  pages = {062306},
  year = {2010},
  doi = {10.1103/PhysRevA.82.062306},
}

@book{breuer2002theory,
  title={The Theory of Open Quantum Systems},
  author={Breuer, H.-P. and Petruccione, F. and others},
  year={2002},
  publisher={Oxford University Press on Demand}
}

@article{buvca2022algebraic,
  title={Algebraic theory of quantum synchronization and limit cycles under dissipation},
  author={Bu{\v{c}}a, B. and Booker, C. and Jaksch, D.},
  journal={SciPost Phys.},
  volume={12},
  pages={097},
  year={2022}
}

@article{KLV,
  title = {Theory of Quantum Error Correction for General Noise},
  author = {Knill, E. and Laflamme, R. and Viola, L.},
  journal = {Phys. Rev. Lett.},
  volume = {84},
  pages = {2525},
  numpages = {0},
  year = {2000},
  doi = {10.1103/PhysRevLett.84.2525},
}

@article{lindblad1976generators,
  title={On the generators of quantum dynamical semigroups},
  author={Lindblad, G.},
  journal={Commun. Math. Phys.},
  volume={48},
  pages={119},
  year={1976},
  publisher={Springer}
}

@article{Gorini:1975nb,
    author = "Gorini, V. and Kossakowski, A. and Sudarshan, E. C. G.",
    title = "{Completely Positive Dynamical Semigroups of $N$ Level Systems}",
    doi = "10.1063/1.522979",
    journal = "J. Math. Phys.",
    volume = "17",
    pages = "821",
    year = "1976"
}

@article{grigoletto2023modelreductionquantumsystems,
    title={Model Reduction for Quantum Systems: Discrete-time Quantum Walks and Open {M}arkov Dynamics},
    author={Grigoletto, T. and Ticozzi, F.},
    journal={IEEE Trans. Inf. Theory},
    year={2025},
    pages ={8524},
    volume ={71}, 
    doi={10.1109/TIT.2025.3601118},
    publisher={IEEE},
}

@article{grigoletto2024exactmodelreductioncontinuoustime,
      title={Exact Model Reduction for Continuous-Time Open Quantum Dynamics}, 
      author={T. Grigoletto and Y.  Tao and F.  Ticozzi and L. Viola},
      year={2025},
      journal ={Quantum},
      volume ={9},
      pages={1814},
      doi={https://doi.org/10.22331/q-2025-07-29-1814}
  }

@article{Buca_2012,
doi = {10.1088/1367-2630/14/7/073007},
year = {2012},
volume = {14},
pages = {073007},
author = {B. Bu\ifmmode \check{c}\else \v{c}\fi{}a and T. Prosen},
title = {A note on symmetry reductions of the {L}indblad equation: transport in constrained open spin chains},
journal = {New J. Phys.},
}

@article{Derrida,
title = {An exactly soluble non-equilibrium system: {T}he asymmetric simple exclusion process},
journal = {Phys. Rep.},
volume = {301},
pages = {65},
year = {1998},
doi = {https://doi.org/10.1016/S0370-1573(98)00006-4},
author = {B. Derrida},
}

@article{Iemini2018,
  title = {Boundary Time Crystals},
  author = {Iemini, F. and Russomanno, A. and Keeling, J. and Schir\`o, M. and Dalmonte, M. and Fazio, R.},
  journal = {Phys. Rev. Lett.},
  volume = {121},
  pages = {035301},
  year = {2018},
  doi = {10.1103/PhysRevLett.121.035301},
}

@article{Yang2025,
	author = {Yang, S. and Wang, Z. and Fu, L. and Jie, J.},
	doi = {10.1038/s42005-025-02040-1},
	journal = {Commun. Phys.},
	pages = {114},
	title = {Emergent continuous time crystal in dissipative quantum spin system without driving},
	volume = {8},
	year = {2025},
}

@article{mondkar2026,
      title={Dynamical Quantum Phase Transitions in Boundary Time Crystals}, 
      author={S. Mondkar and P. Ghosh and U. Sen},
      year={2026},
      journal={arXiv:2602.04792},
      archivePrefix={arXiv},
      primaryClass={quant-ph},
      url={https://arxiv.org/abs/2602.04792}, 
}

@article{Bruder2025,
  title = {Quantum Synchronization of Twin Limit-Cycle Oscillators},
  author = {Kehrer, T. and Bruder, C. and Solanki, P.},
  journal = {Phys. Rev. Lett.},
  volume = {135},
  pages = {063601},
  year = {2025},
  doi = {10.1103/f97j-474c},
}

@article{PhysRevA.89.022118,
  title = {Symmetries and conserved quantities in {L}indblad master equations},
  author = {Albert, V. V. and Jiang, L.},
  journal = {Phys. Rev. A},
  volume = {89},
  pages = {022118},
  year = {2014},
  doi = {10.1103/PhysRevA.89.022118},
}

@article{azouit2017towards,
  title={Towards generic adiabatic elimination for bipartite open quantum systems},
  author={Azouit, R. and Chittaro, F. and Sarlette, A. and Rouchon, P.},
  journal={Quantum Sci. Tech.},
  volume={2},
  pages={044011},
  year={2017},
  publisher={IOP Publishing}
}

@article{Robin1,
  title={Convergence of bipartite open quantum systems stabilized by reservoir engineering}, 
  author={R. Robin and P. Rouchon and L.-A. Sellem},
  journal={Ann. Henri Poincar\'e},
  year={2025},
  volume={26},
  doi = {https://doi.org/10.1007/s00023-024-01481-8},
  pages={1769}, 
}

@article{Robin2,
  title={Unconditionally stable time discretization of {L}indblad master equations in infinite dimension using quantum channels}, 
  author={R. Robin and P. Rouchon and L.-A. Sellem},
  journal={arXiv:2503.01712},
  year={2025},
  url={https://doi.org/10.48550/arXiv.2503.01712}, 
}

@article{vdP1,
author = {Y. Li  and others},
title = {Experimental realization and synchronization of a quantum van der {P}ol oscillator},
journal = {Science Adv.},
volume = {11},
pages = {eady5649},
year = {2025},
doi = {10.1126/sciadv.ady5649},
}

@article{vdP2,
author = {J. Liu and Q. Wu and J. E. Moore and H. Haeffner and C. W. W\"{a}chtler },
title = {Observation of synchronization between two quantum van der {P}ol oscillators in trapped ions}, 
journal = {Phys. Rev. X},
volume={16}, 
pages ={021062}, 
year = {2026},
doi = {https://doi.org/10.1103/w1bm-wjl4},
}

@article{Terhal_2020,
doi = {10.1088/2058-9565/ab98a5},
year = {2020},
volume = {5},
pages = {043001},
author = {Terhal, B. M. and Conrad, J. and Vuillot, C.},
title = {Towards scalable bosonic quantum error correction},
journal = {Quantum Sci. Tech.},
}

@article{accardi1982quantum,
  title={Quantum stochastic processes},
  author={Accardi, L. and Frigerio, A. and Lewis, J. T.},
  journal={Publications of the Research Institute for Mathematical Sciences},
  volume={18},
  pages={97},
  year={1982},
  publisher={Research Institute forMathematical Sciences}
}

@article{Houck,
  title = {Controlling the Spontaneous Emission of a Superconducting Transmon Qubit},
  author = {Houck, A. A. and Schreier, J. A. and Johnson, B. R. and Chow, J. M. and Koch, Jens and Gambetta, J. M. and Schuster, D. I. and Frunzio, L. and Devoret, M. H. and Girvin, S. M. and Schoelkopf, R. J.},
  journal = {Phys. Rev. Lett.},
  volume = {101},
  pages = {080502},
  year = {2008},
  doi = {10.1103/PhysRevLett.101.080502},
}

@article{PhysRevA.109.062206,
  title = {Complete positivity violation of the reduced dynamics in higher-order quantum adiabatic elimination},
  author = {Tokieda, M. and Elouard, C. and Sarlette, A. and Rouchon, P.},
  journal = {Phys. Rev. A},
  volume = {109},
  pages = {062206},
  year = {2024},
  doi = {10.1103/PhysRevA.109.062206},
}

@article{wolf2010inverseeigenvalueproblemquantum,
      title={The inverse eigenvalue problem for quantum channels}, 
      author={M. M. Wolf and D. Perez-Garcia},
      year={2010},
      journal={arXiv:1005.4545},
      archivePrefix={arXiv},
      primaryClass={quant-ph},
      url={https://arxiv.org/abs/1005.4545}, 
}

@book{kato2013perturbation,
  title={Perturbation Theory for Linear Operators},
  author={Kato, T.},
  volume={132},
  year={2013},
  publisher={Springer Science \& Business Media}
}

@article{lindblad1999general,
  title={A general no-cloning theorem},
  author={Lindblad, G.},
  journal={Lett. Math. Phys.},
  volume={47},
  pages={189},
  year={1999},
  publisher={Springer}
}

@article{bialonczyk2018application,
  title={Application of {S}hemesh theorem to quantum channels},
  author={Bia{\l}o{\'n}czyk, M. and Jamio{\l}kowski, A. and {\.Z}yczkowski, K.},
  journal={J. Math. Phys.},
  volume={59},
  pages ={102204},
  year={2018},
  publisher={AIP Publishing}
}

@article{grigoletto2025quantummodelreductioncontinuoustime,
title={Quantum model reduction for continuous-time quantum filters}, 
      author={T. Grigoletto and C. Pellegrini and F. Ticozzi},
      year={2025},
journal = {Ann. Henri Poincar\'e},
doi={10.1007/s00023-025-01622-7},
volume = {},
pages ={}, 
}

@article{PhysRevLett.113.240406,
  title = {Coherent Quantum Dynamics in Steady-State Manifolds of Strongly Dissipative Systems},
  author = {Zanardi, P. and Campos Venuti, L.},
  journal = {Phys. Rev. Lett.},
  volume = {113},
  pages = {240406},
  year = {2014},
  doi = {10.1103/PhysRevLett.113.240406},
}

@article{PhysRevResearch.5.L012003,
  title = {Symmetry-induced decoherence-free subspaces},
  author = {Dubois, J. and Saalmann, U. and Rost, J. M.},
  journal = {Phys. Rev. Res.},
  volume = {5},
  pages = {L012003},
  year = {2023},
  doi = {10.1103/PhysRevResearch.5.L012003},
}

@book{franchini2017introduction,
  title={An Introduction to Integrable Techniques for One-dimensional Quantum Systems},
  author={Franchini, F.},
  volume={940},
  year={2017},
  publisher={Springer, Berlin}
}

@phdthesis{zhang2024driven,
  title={Driven-dissipative Phase Transitions for {M}arkovian Open Quantum Systems},
  author={Zhang, Y.},
  year={2024},
  school={Duke University}
}

@article{zanardi2015geometry,
  title={Geometry, robustness, and emerging unitarity in dissipation-projected dynamics},
  author={Zanardi, P. and Campos Venuti, L.},
  journal={Phys. Rev. A},
  volume={91},
  pages={052324},
  year={2015},
  publisher={APS}
}

@article{zanardi2016dissipative,
  title={Dissipative universal {L}indbladian simulation},
  author={Zanardi, P. and Marshall, J. and Campos Venuti, L.},
  journal={Phys. Rev. A},
  volume={93},
  pages={022312},
  year={2016},
  publisher={APS}
}

@article{buvca2019non,
  title={Non-stationary coherent quantum many-body dynamics through dissipation},
  author={Bu{\v{c}}a, B. and Tindall, J. and Jaksch, D.},
  journal={Nat. Commun.},
  volume={10},
  pages={1730},
  year={2019},
  publisher={Nature Publishing Group UK London}
}

@article{Tindall2020,
doi = {10.1088/1367-2630/ab60f5},
year = {2020},
volume = {22},
pages = {013026},
author = {Tindall, J. and Sanchez-Munoz, C. and Bu{\v{c}}a, B. and Jaksch, D.},
title = {Quantum synchronisation enabled by dynamical symmetries and dissipation},
journal = {New J. Phys.},
}

@article{li2024exact,
  title={Exact steady state of quantum van der {P}ol oscillator: critical phenomena and enhanced metrology},
  author={Li, Y. and Zhang, X. and Liu, Y.-C.},
  journal={Phys. Rev. A},
  volume={112}, 
  pages ={L021701},
  doi={https://doi.org/10.1103/zybb-vxfz},
  year={2024}
}

@article{ben2021quantum,
  title={Quantum limit cycles and the {R}ayleigh and van der {P}ol oscillators},
  author={Ben Arosh, L. and Cross, M. C. and Lifshitz, R.},
  journal={Phys. Rev. Res.},
  volume={3},
  pages={013130},
  year={2021},
  publisher={APS}
}

@article{gu2024spontaneous,
  title={Spontaneous symmetry breaking in open quantum systems: strong, weak, and strong-to-weak},
  author={Gu, D. and Wang, Z. and Wang, Z.},
  journal={Phys. Rev. B},
  volume ={112}, 
  pages={245123}, 
  year={2024}
}

@article{Zhao2025,
  title={Quantum Synchronization of Perturbed Oscillating Coherences},
  author={Y. J. Zhao and J. E. Moore and J. Thingna and C. W. W\"{a}chtler},
  journal={arXiv:2510.11601},
  year={2025}
}

@INPROCEEDINGS{riva2024cdc,
  author={Riva, A. and Sarlette, A. and Rouchon, P.},
  booktitle={2024 IEEE 63rd Conference on Decision and Control (CDC)}, 
  title={Explicit formulas for adiabatic elimination with fast unitary dynamics}, 
  year={2024},
  volume={},
  number={},
  pages={755},
  doi={10.1109/CDC56724.2024.10886784}}

@INPROCEEDINGS{forni2019cdc,
  author={Forni, P. and Launay, T. and Sarlette, A. and Rouchon, P.},
  booktitle={2019 IEEE 58th Conference on Decision and Control (CDC)}, 
  title={A palette of approaches for adiabatic elimination in bipartite open quantum systems with Hamiltonian dynamics on target}, 
  year={2019},
  volume={},
  number={},
  pages={1362},
  doi={10.1109/CDC40024.2019.9028922}}

@inproceedings{forni2018adiabatic,
  title={Adiabatic elimination for multi-partite open quantum systems with non-trivial zero-order dynamics},
  author={Forni, P. and Sarlette, A. and Capelle, T. and Flurin, E. and Del{\'e}glise, S. and Rouchon, P.},
  booktitle={2018 IEEE Conference on Decision and Control (CDC)},
  pages={6614},
  year={2018},
}

@article{Benoist_2021,
   title={Emergence of Jumps in Quantum Trajectories via Homogenization},
   volume={387},
   DOI={10.1007/s00220-021-04179-8},
   journal={Commun. Math. Phys.},
   publisher={Springer Science and Business Media LLC},
   author={Benoist, T. and Bernardin, C. and Chetrite, R. and Chhaibi, R. and Najnudel, J. and Pellegrini, C.},
   year={2021},
    pages={1821} 
   }

@Book{FB-LNS,
  author =    {F. Bullo},
  title =     {Lectures on Network Systems},
  year =      2024,
  edition =   {{1.7}},
  publisher = {Kindle Direct Publishing},
  ISBN =      {978-1986425643},
  url =       {https://fbullo.github.io/lns},
}

@article{ticozzi2013steadystateentanglementengineeredquasilocal,
      title={Steady-State Entanglement by Engineered Quasi-Local {M}arkovian Dissipation}, 
      author={F. Ticozzi and L. Viola},
      year={2014},
      journal ={Quantum Inf. Comput.}, 
      volume={14}, 
      pages={0265},  
      doi= {https://doi.org/10.26421/QIC14.3-4-5},
}

@article{PhysRevE.102.062210,
  title = {Integrability of one-dimensional {L}indbladians from operator-space fragmentation},
  author = {Essler, F. H. L. and Piroli, L.},
  journal = {Phys. Rev. E},
  volume = {102},
  pages = {062210},
  year = {2020},
  doi = {10.1103/PhysRevE.102.062210},
}

@article{PhysRevB.110.104303,
  title = {Oscillating-mode gap: An indicator of phase transitions in open quantum many-body systems},
  author = {Haga, T.},
  journal = {Phys. Rev. B},
  volume = {110},
  pages = {104303},
  year = {2024},
  doi = {10.1103/PhysRevB.110.104303},
}

@article{PhysRevA.111.052206,
  title = {Time-convolutionless master equation applied to adiabatic elimination},
  author = {Tokieda, M. and Riva, A.},
  journal = {Phys. Rev. A},
  volume = {111},
  pages = {052206},
  year = {2025},
  doi = {10.1103/PhysRevA.111.052206},
}

@article{Arenz_2020,
   title={Emerging unitary evolutions in dissipatively coupled systems},
   volume={101},
   pages ={ 022101},
   DOI={10.1103/physreva.101.022101},
   journal={Phys. Rev. A},
   author={Arenz, C. and Metelmann, A.},
   year={2020},
}

@article{loizeau2026krylovspaceperturbationtheory,
      title={Krylov space perturbation theory for quantum synchronization in closed systems}, 
      author={N. Loizeau and B. Buča},
      year={2026},
      journal={arXiv:2602.11431},
      archivePrefix={arXiv},
      primaryClass={cond-mat.dis-nn},
      url={https://arxiv.org/abs/2602.11431}, 
}

@article{roberts1989appropriate, 
    title={Appropriate initial conditions for asymptotic descriptions of the long term evolution of dynamical systems}, 
    volume={31}, 
    DOI={10.1017/S0334270000006470}, 
    number={1}, 
    journal={The Journal of the Australian Mathematical Society. Series B. Applied Mathematics}, author={Roberts, A. J.}, 
    year={1989}, 
    pages={48–75}
}

@article{roberts2015macroscale,
  title={Macroscale, slowly varying, models emerge from the microscale dynamics},
  author={Roberts, A. J.},
  journal={IMA J. Appl. Math.},
  volume={80},
  pages={1492},
  year={2015},
  publisher={Oxford University Press}
}

@article{szyld2006many,
  title={The many proofs of an identity on the norm of oblique projections},
  author={Szyld, D. B.},
  journal={Num. Alg.},
  volume={42},
  pages={309},
  year={2006},
  publisher={Springer}
}

@article{Tamascelli_2018,
   title={Nonperturbative Treatment of non-{M}arkovian Dynamics of Open Quantum Systems},
   volume={120},
   DOI={10.1103/physrevlett.120.030402},
   journal={Phys. Rev. Lett.},
   publisher={American Physical Society (APS)},
   author={Tamascelli, D. and Smirne, A. and Huelga, S. F. and Plenio, M. B.},
   year={2018},
}

@book{lidar2013quantum,
  title={Quantum Error Correction},
  author={Lidar, D. A. and Brun, T. A.},
  year={2013},
  publisher={Cambridge University Press}
}

\end{document}